\documentclass[10pt,3p]{elsarticle}
\journal{arXiv}

\usepackage{amssymb}
\usepackage[fleqn]{amsmath}
\usepackage{amsthm}
\usepackage[hidelinks]{hyperref}
\usepackage{cancel}
\usepackage{nicefrac}
\usepackage{xcolor}
\usepackage{subcaption}

\newcommand{\bs}[1]{\boldsymbol{#1}}
\newcommand{\reals}{\mathbb{R}}
\newcommand{\dims}{n_{\text{dim}}}
\newcommand{\lhsop}{\mathcal{N}}
\newcommand{\rhsop}{\mathcal{L}}
\newcommand{\residop}{\mathcal{R}}
\newcommand{\qoi}{\mathcal{J}}

\newcommand{\fspace}{\mathcal{V}}
\newcommand{\error}{\mathcal{E}}

\theoremstyle{definition}

\newtheorem{prop}{Proposition}

\begin{document}

\begin{frontmatter}

\title{Linearization Errors in Discrete Goal-Oriented Error Estimation}
\author[sandia]{Brian N. Granzow\corref{correspondence}}
\author[sandia]{D. Thomas Seidl}
\author[sandia]{Stephen D. Bond}
\address[sandia]{Sandia National Laboratories \\
P.O. Box 5800 \\
Albuquerque, NM 87185-1321}
\cortext[correspondence]{Corresponding author, bngranz@sandia.gov}


\begin{abstract}
This paper is concerned with goal-oriented \emph{a posteriori} error
estimation for nonlinear functionals in the context of nonlinear variational
problems solved with continuous Galerkin finite element discretizations. A
two-level, or discrete, adjoint-based approach for error estimation is
considered. The traditional method to derive an error estimate in this context
requires linearizing both the nonlinear variational form and the nonlinear
functional of interest which introduces linearization errors into the error
estimate. In this paper, we investigate these linearization errors. In 
particular, we develop a novel discrete goal-oriented error estimate that
accounts for traditionally neglected nonlinear terms at the expense of greater
computational cost. We demonstrate how this error estimate can be used to drive
mesh adaptivity. We show that accounting for linearization errors in the error
estimate can improve its effectivity for several nonlinear model problems and
quantities of interest. We also demonstrate that an adaptive strategy based on
the newly proposed estimate can lead to more accurate approximations of the
nonlinear functional with fewer degrees of freedom when compared to uniform
refinement and traditional adjoint-based approaches.
\end{abstract}

\end{frontmatter}

\section{Introduction}
\label{sec:introduction}

Finite element simulations have become ubiquitous in engineering practice.
It is natural to question the accuracy of a given finite element solution
and, in particular, the accuracy of post-processed output quantities upon
which critical design decisions may be made. For instance, an aerospace
engineer may be interested in accurately assessing the drag over an airfoil
design. \emph{A posteriori} error estimation provides a mechanism to
approximate discretization errors in this context, and, in particular,
goal-oriented \emph{a posteriori} error estimation can approximate the
discretization error in a chosen functional quantity of interest (QoI).
Localized information obtained from a goal-oriented error estimate
can then be used to control the QoI discretization error by adapting the finite
element mesh via local mesh modifications \cite{li20053d, alauzet2006parallel}.

Goal-oriented \emph{a posteriori} error estimation has been successfully
applied in a wide variety of contexts over the past thirty
years~\cite{peraire1998bounds, prudhomme1999goal, venditti2000adjoint,
venditti2003anisotropic,  becker2001optimal, fidkowski2011review,
aksoylu2011pbe, holst2015convergence, granzow2017output}. Presently,
we consider \emph{two-level} or \emph{discrete} approaches for goal-oriented
error estimation. Modern finite element software
packages \cite{dealII, mfem, pumi, fenics} significantly ease the burden of
implementing a two-level adjoint-based error estimation scheme which may,
in part, explain why such schemes have been applied in a variety of contexts
recently \cite{ali2017optimal, ahrabi2017adjoint, blonigan2023evaluation,
donoghue2021spatio, fidkowski2022output, granzow2018adjoint, sleeman2022goal}.
Presently, we have used the software PUMI \cite{pumi} to implement parallel,
adaptive finite element simulations driven by goal-oriented error estimation.

Inherent in the traditional derivation of a discrete goal-oriented error estimate
are linearizations of the PDE residual and of the QoI. Once linearized,
the higher-order terms are subsequently disregarded in
typical discrete goal-oriented error estimate
derivations~\cite{fidkowski2011review}. We refer to these neglected higher-order
terms as \emph{linearization errors}. The purpose of this paper is to
investigate the effect of including these linearization errors in a
goal-oriented error estimate and the effect that localizing these
linearization errors has on mesh adaptation.

We briefly review the relevant portions of the literature that concern
linearization errors in the context of goal-oriented error estimation.
Previous authors have included the adjoint residual in the derivation of
goal-oriented error estimates to obtain an error estimate that is
\emph{locally} third-order, but \emph{globally} second-order~\cite{
venditti2000adjoint, becker2001optimal, kast2017introduction}. Additionally,
Venditti and Darmofal~\cite{venditti2000adjoint} quantified the error due to
nonlinear effects when using a two-level approach, but did not include these
effects in their error estimate. {\c{S}}im{\c{s}}ek et al.~\cite{
csimcsek2015duality} explicitly accounted for the residual linearization error
in the derivation of a two-level goal-oriented error
estimate for various time-dependent PDEs, but did not consider nonlinear QoIs
and did not consider the \emph{localization} of higher-order effects to drive
mesh adaptivity. Finally, more recently Endtmayer et
al.~\cite{endtmayer2020two} and Dolejsi and Congreve~\cite{dolejsi2023goaloriented}
developed methods to incorporate linearization errors (as well as errors that
arise from iteratively solving nonlinear algebraic equations) in a variational
setting that depend on the strong-form of the adjoint residual.

The contributions of this paper can be stated as follows. First, by considering
linearization errors, we derive a novel two-level adjoint-based error estimate
that exactly represents the QoI discretization error between two function
spaces and is independent of any form of the adjoint residual. Second, we
prove that the solution of a requisite nonlinear scalar problem for this
newly proposed estimate can be known \emph{a priori} when quadratic QoIs
are considered. Third, we outline a convenient adjoint verification
procedure that makes use of linearization errors and we highlight
how this procedure can be applied to existing adjoint-based error
estimation codes. We then present an adaptive mesh scheme that
utilizes the newly proposed error estimate. Last, we demonstrate that the
novel error estimate can lead to more accurate error estimates and more optimal
meshes in certain scenarios when compared to a traditional
adjoint-weighted residual error estimate.

The remainder of this paper is structured as follows. First, we provide
a review of discrete goal-oriented error estimation in a general setting and
motivate our interest in studying linearization errors in this context.
Next, we introduce and derive a novel discrete adjoint-based error estimate
that is free of linearization errors. The evaluation of this estimate
requires the solution of an auxiliary nonlinear scalar problem. We then
provide the solution to this nonlinear scalar problem when quadratic
QoIs are considered. Next, we highlight that linearization errors can be
used to perform adjoint verification in a traditional adjoint-weighted residual
error estimation context. We then discuss the localization of the newly
proposed error estimate and how this localization can be used to adaptively
modify meshes to reduce QoI discretization errors. Next, we investigate
a nonlinear model problem with multiple exemplar QoIs
to verify and demonstrate the effectiveness of the newly proposed estimate.
We conclude by summarizing our findings and proposing avenues for
future work.

\section{Review of Discrete Goal-Oriented Error Estimation}
\label{sec:review}

\subsection{Two-Level Strategy}
\label{ssec:two_level}

Let $\Omega \subset \reals^{\dims}$ denote an open bounded domain in $\dims$
spatial dimensions with a Lipschitz continuous boundary $\Gamma$. Consider the
stationary nonlinear variational problem: find $u \in \fspace$ such that
\begin{equation}
\lhsop(u; w) = \rhsop(w), \quad \forall \, w \in \fspace.
\label{eq:weak_form}
\end{equation}
Here, $\fspace$ denotes an appropriate Sobolev space for the underlying PDE,
$\lhsop : \fspace \times \fspace \to \reals$ denotes a semilinear form,
nonlinear in its first argument and linear in its second, $\rhsop: \fspace \to
\reals$ denotes a linear functional, $u$ denotes the exact weak solution to the
chosen PDE, and $w$ denotes an arbitrary weighting function. Let $\cup_{e=1}^
{n_{el}} \overline{\Omega^e} = \overline{\Omega}$ denote a partitioning of the
domain $\Omega$ into $n_{el}$ non-overlapping elements,
such that $\Omega^i \cap \Omega^j = \varnothing$. Let
$\fspace^H \subset \fspace$ denote a finite-dimensional
space. A finite element formulation of the variational 
problem~\eqref{eq:weak_form} can be stated: find $u^H \in \fspace^H$ such that
\begin{equation}
\lhsop(u^H; w^H) = \rhsop(w^H), \quad \forall \, w^H \in \fspace^H,
\label{eq:fem_form}
\end{equation}
which results in a system of $N$ nonlinear algebraic equations $\bs{R}^H :
\reals^N \to \reals^N$
\begin{equation}
\bs{R}^H(\bs{u}^H) = \bs{0},
\label{eq:primal_coarse}
\end{equation}
where $\bs{u}^H \in \reals^N$ denotes the vector of nodal solution
coefficients. The problem posed by equation~\eqref{eq:primal_coarse} is
referred to as the \emph{primal problem}. The primal 
problem~\eqref{eq:primal_coarse} is solved via Newton's method, where 
iterations of the steps:
\begin{equation}
\begin{aligned}
\left[ \frac{\partial \bs{R}^H}{\partial \bs{u}^H} \biggr|_{\bs{u}^H_i} \right]
\, \delta \bs{u}^H_i = -\bs{R}^H(\bs{u}^H_i), \\
\bs{u}^H_{i+1} = \bs{u}^H_i + \delta \bs{u}^H_i,
\label{eq:primal_newton}
\end{aligned}
\end{equation}
are performed until the convergence criterion $\| \bs{R}^H(\bs{u}^H) \|_2 <
\text{TOL}$ is satisfied. Here, $\text{TOL}$ is some user-specified
convergence tolerance, $\bs{u}^H_i$ denotes the solution vector at the $i$th
Newton iteration, and $\delta \bs{u}^H_i$ denotes the incremental solution
update at the $i$th iteration.

Let $\qoi(u): \fspace \to \reals$ denote a functional that is of particular
interest to a given simulation, referred to as the \emph{quantity of
interest}. The aim of goal-oriented error estimation is to approximate the
discretization error in the QoI, $\error := \qoi(u) - \qoi(u^H)$.
To this end, we adopt a two-level error estimation 
strategy~\cite{venditti2000adjoint,venditti2002grid,venditti2003anisotropic,
fidkowski2011review}, which is sometimes referred to as a \emph{discrete}
approach for goal-oriented error estimation. For this strategy, a second
finite-dimensional space $\fspace^h$ is chosen, such that $\fspace^H \subset
\fspace^h \subset \fspace$. In this context, $\fspace^H$ and $\fspace^h$ are
referred to the \emph{coarse space} and \emph{fine space}, respectively.

The finite element problem~\eqref{eq:fem_form} posed on the fine space
$\fspace^h$ leads to a system of $n$ nonlinear algebraic equations $\bs{R}^h :
\reals^n \to \reals^n$, where $n > N$, and
\begin{equation}
\bs{R}^h(\bs{u}^h) = \bs{0}.
\label{eq:primal_fine}
\end{equation}
Here $\bs{u}^h$ denotes the solution vector of nodal coefficients for the fine
problem. Additionally, the QoI can be evaluated using the coarse and fine
space discretizations, which we denote as $\qoi^H : \reals^N \to \reals$ and
$\qoi^h : \reals^n \to \reals$, respectively. Let $\bs{I}^H_h : \fspace^H \to
\fspace^h$ denote the prolongation operator from the coarse space to the fine
space defined by interpolation and let $(\cdot)^H_h := \bs{I}^H_h (\cdot)^H$
denote the prolongation of a coarse space vector $(\cdot)^H$ onto the fine
space $\fspace^h$. Because the exact error $\error$ is generally unknowable,
two-level approaches for goal-oriented error estimation introduce computable
approximations for the QoI discretization error $\error^h$ between the coarse
and fine space, where
\begin{equation}
\begin{aligned}
\error^h &:= \qoi^h(\bs{u}^h) - \qoi^H(\bs{u}^H) \\
&= \qoi^h(\bs{u}^h) - \qoi^h(\bs{u}^H_h).
\end{aligned}
\label{eq:discrete_qoi_error}
\end{equation}
Here the error $\error^h$ serves as a proxy for the true discretization error
$\error$.

\subsection{Choices for the Fine Space}
\label{ssec:fine_space}

There exist several common choices for the fine space $\fspace^h$. Of primary
importance in making this choice is that $\error^h$ is a \emph{good} proxy for
the true QoI discretization error $\error$. With this in mind, the fine space
could be obtained by uniformly refining the mesh used for the primal problem.
We refer to this approach as $h$-enrichment. Alternatively, the fine space
could be defined by using a finite element basis with an increased polynomial
order when compared to the coarse space basis. We refer to this approach as
$p$-enrichment. Lastly, a combination of both $h$- and $p$-enrichment could be
used to define the fine space, which we refer to as $hp$-enrichment. As a note,
more esoteric approaches can be chosen for the fine space that we will not
consider in the present context~\cite{granzow2017output,
richter2015variational}.

The choice of the fine space affects the asymptotic behavior of the error
estimate $\error^h$. For example, consider a QoI that converges at a rate of
$k$ on the coarse space and at a rate of $k+l$ on the fine space, such that
$\qoi(u) - \qoi^h(\bs{u}^H_h) = c | \Omega^e_H |^k$ and $\qoi(u) -
\qoi^h(\bs{u}^h) = c | \Omega^e_h |^{k+l}$. Here, $|\Omega^e_H| :=
\text{meas}(\Omega^e_H)$ and $|\Omega^e_h| := \text{meas}(\Omega^e_h)$ denote
the characteristic mesh size on the coarse space and fine space, respectively.
The ratio of the error in the QoI between the two spaces $\error^h$ and the
true discretization error $\error$ can then be written as
\begin{equation}
\begin{aligned}
\frac{\error^h}{\error} &= \frac{\left( \qoi(u) - \qoi^h(\bs{u}^H_h) \right) -
\left( \qoi(u) - \qoi^h(\bs{u}^h) \right)}{\qoi(u) - \qoi^h(\bs{u}^H_h)} \\
&= \frac{c \, |\Omega^e_H|^k - c \, |\Omega^e_h|^{k+l}}{c \, |\Omega^e_H|^k} \\
&= 1 - \left[ \frac{|\Omega^e_h|}{|\Omega^e_H|} \right]^k |\Omega^e_h|^l.
\label{eq:asymptotic_effectivity_h}
\end{aligned}
\end{equation}
When $p$-enrichment is utilized, $l>0$ and the characteristic mesh size on
both spaces is equal, with $|\Omega^e_H| = |\Omega^e_h|$. This results in
the asymptotic behavior $\error^h/\error \to 1$ as $|\Omega^e_H| \to 0$. When
uniform $h$-enrichment is utilized, the ratio of characteristic mesh sizes is
given as $|\Omega^e_h|/|\Omega^e_H| = 1/2$ and there is no increase in the
convergence rate of the QoI, with $l = 0$. This results in the asymptotic
behavior $\error^h/\error \to 1 - (1/2)^k$ as $|\Omega^e_H| \to 0$.
Presently, we choose the fine space to be $p$-enrichment, though this is not a
requirement for the following discussions.

\subsection{Traditional Adjoint-Weighted Residual Estimate}
\label{ssec:adjoint_weighted_residual}

Given a choice for the fine space and the notation introduced in
Section~\ref{ssec:two_level}, the residual equations on the fine space can be 
expanded about the prolonged coarse solution $\bs{u}^H_h$ as
\begin{equation}
\cancelto{0}{\bs{R}^h(\bs{u}^h)} = \bs{R}^h(\bs{u}^H_h) + \left[ \frac{
\partial \bs{R}^h}{\partial \bs{u}^h} \biggr|_{\bs{u}^H_h} \right] \bs{e}^h +
\bs{E}^{\residop}_L.
\label{eq:resid_taylor}
\end{equation}
Here the left-hand side vanishes by virtue of the governing
equations on the fine space~\eqref{eq:primal_fine}, $\bs{e}^h := \left(
\bs{u}^h - \bs{u}^H_h \right)$ denotes the solution discretization error
between the coarse and fine space, and $\bs{E}^{\residop}_L$ denotes
higher-order terms in the residual Taylor expansion that we will refer to as
the \emph{residual linearization error}. Similarly, the QoI can be expanded
about the prolonged coarse solution in a Taylor series as
\begin{equation}
\qoi^h(\bs{u}^h) = \qoi^h(\bs{u}^H_h) + \left[ \frac{\partial \qoi^h}{\partial
\bs{u}^h} \biggr|_{\bs{u}^H_h} \right] \bs{e}^h + \error^{\qoi}_L,
\label{eq:qoi_taylor}
\end{equation}
where $\error^{\qoi}_L$ denotes higher-order terms  that we will refer to as
the \emph{QoI linearization error}. Neglecting the residual linearization
error in the expansion~\eqref{eq:resid_taylor}, the discretization error can
be approximated linearly as
\begin{equation}
\bs{e}^h \approx \bs{e}^h_L := - \left[ \frac{\partial \bs{R}^h}{\partial
\bs{u}^h} \biggr|_{\bs{u}^H_h} \right]^{-1} \bs{R}^h(\bs{u}^H_h).
\label{eq:linearized_error}
\end{equation}
Let $\bs{z}^h$ denote the solution to the so-called \emph{adjoint problem},
given as
\begin{equation}
\left[ \frac{\partial \bs{R}^h}{\partial \bs{u}^h} \biggr|_{\bs{u}^H_h}
\right]^T \bs{z}^h = \left[ \frac{\partial \qoi^h}{\partial \bs{u}^h}
\biggr|_{\bs{u}^H_h} \right]^T.
\label{eq:adjoint_fine}
\end{equation}
Substituting the linearized approximation for the discretization 
error~\eqref{eq:linearized_error} into the QoI Taylor 
expansion~\eqref{eq:qoi_taylor}, neglecting the QoI linearization 
error, and utilizing the adjoint solution~\eqref{eq:adjoint_fine} 
yields the well-known adjoint-weighted residual error estimate:
\begin{equation}
\error^h \approx \eta_1 := -\bs{z}^h \cdot \bs{R}^h(\bs{u}^h_H)
\label{eq:adjoint_weighted_residual}.
\end{equation}

The introduction of the adjoint problem is motivated by two considerations.
First, it replaces a potentially cost prohibitive nonlinear primal solve on
the fine space with a cheaper linear solve on the fine space. Depending on the
nonlinearity of the primal PDE and the choice of fine space, the cost of
solving the adjoint problem~\eqref{eq:adjoint_fine} on the fine space can be
comparable or cheaper than the cost of solving the primal problem on the coarse
space~\eqref{eq:primal_coarse}. Second, and perhaps most important, the
adjoint solution itself provides information about the spatial distribution of
contributions to the QoI discretization error. Concretely, this is because
$\bs{z}^h$ represents the linear sensitivity of the QoI with respect to
perturbations in the PDE residual.

For example, consider a point-wise QoI inside of the
computational domain. We could na\"{i}vely solve for the linearized
discretization error~\eqref{eq:linearized_error} for the purpose of replacing
a nonlinear solve on the fine space with a linear solve on the fine space.
We could then use the linearized discretization error $\bs{e}^h_L$ to
approximate the fine space solution as $\bs{u}^h \approx \bs{u}^H_h +
\bs{e}^h_L$ and, in turn,  approximate the QoI error $\error^h$. In fact, this
process would yield an identical error estimate to $\eta_1$. However, in doing
so, we would obtain no information about the spatial distribution of the QoI
error in the domain except at the point of interest and would thus have limited
tools to control the discretization error in the QoI via mesh adaptivity.

\subsection{Motivation For Studying Linearization Errors}
\label{ssec:motivation}

Much has been said about the error estimate $\eta_1$~\cite{fidkowski2011review}.
It is, however, less common to consider the effect of higher-order information
when deriving goal-oriented error estimates~\cite{csimcsek2015duality,
endtmayer2020two, dolejsi2023goaloriented}, particularly in a discrete adjoint
setting. This is the primary aim of the present work. Consider, for example,
the function $f(x) = x^2$, whose Taylor expansion about the point $0$ can be
expressed exactly as $f(x) = f(0) + f'(0)x + \frac12f''(0)x^2$. At any
evaluation point $a$, the entire Taylor series expansion is encoded in the
second derivative term, with the linearized approximation $f(a) - f(0)
\approx f'(0)a$ providing no meaningful information about the difference $f(a) -
f(0)$ as it simply evaluates to $0$. Unfortunately, this is exactly the
approximation that is made in the derivation of the traditional
adjoint-weighted residual estimate. If
instead we include higher-order information in the function's Taylor expansion,
where $f(a) - f(0) = f'(0)a + f''(0)a^2$, we recognize that we will
recover the exact difference $f(0) - f(a)$.

Further, let $u^H$ denote the function corresponding to the nodal coefficient
vector $\bs{u}^H$ and let $e := u - u^H$ denote the exact discretization
error. Consider the quadratic functional $\qoi(u) = \int_{\Omega} \nabla u
\cdot \nabla u \; \text{d} \Omega$. The discretization error in this functional
can be exactly represented as
\begin{equation}
\qoi(u) - \qoi(u^H) =
2 \int_{\Omega} \nabla u^H \cdot \nabla e \; \text{d} \Omega +
\int_{\Omega} \nabla e \cdot \nabla e \; \text{d} \Omega.
\label{eq:example_qoi_error}
\end{equation}
Here, we recognize that the first integral on the right-hand side of
the error expression \eqref{eq:example_qoi_error} corresponds to the
linearization used for the QoI in the traditional adjoint-weighted residual
approach and the second integral corresponds to the linearization error
introduced in the estimate. Naturally, for a convergent discretization
$e \to 0$ as $H \to 0$, and one would expect the significance of the
linearization error to quickly diminish as $H \to 0$ since it is
$\mathcal{O}(e^2)$. However, the integrand
in the linearization error term is strictly positive, while the first integral
has the potential to tend to zero quickly relative to the linearization
error due to subtractive cancellation. In this scenario, if the linearization
error is neglected then there could be a significant under-prediction of the
actual error.

These are our primary motivations for studying linearization errors and
including higher-order information in the derivation of a goal-oriented error
estimate. While the above scenarios may seem pathological, we provide concrete
examples where the traditional error estimate $\eta_1$ performs poorly for
nonlinear QoIs in Section~\ref{sec:results}. While we are presently
concerned with extending discrete adjoint-weighted error estimates to include
nonlinear effects, it is also worth noting that the reliability of the
traditional dual weighted residual error \cite{becker2001optimal} estimate,
which is derived in a continuous variational setting, has also been
demonstrated to significantly under-predict errors
\cite{nochetto2009safeguarded} in specific contexts.

\section{An Estimate Free of Linearization Errors}
\label{sec:estimate1}

\subsection{Derivation}
\label{ssec:estimate1_derivation}

We can entirely account for the linearization errors $\bs{E}^{
\residop}_L$ and $\error^{\qoi}_L$ in the derivation of a two-level
adjoint-based error estimate. From the mean-value theorem, there exists
a vector $\bs{u}^*$ where the QoI linearization error
$\error^{\qoi}_L$ vanishes, such that
\begin{equation}
\qoi^h(\bs{u}^h) = \qoi^h(\bs{u}^H_h) + \left[ \frac{\partial \qoi^h}{
\partial \bs{u}^h} \biggr|_{\bs{u}^*} \right] \bs{e}^h.
\label{eq:qoi_taylor_exact}
\end{equation}
Here $\bs{u}^*$ denotes a point on a linear path between the coarse-space solution
$\bs{u}^H_h$ and the fine-space solution $\bs{u}^h$, written as
\begin{equation}
\bs{u}^*(\theta) = \bs{u}^H_h + \theta \bs{e}^h, \quad \theta \in [0,1].
\end{equation}
Finding the appropriate value of $\theta$, and thus $\bs{u}^*$, amounts to
solving the nonlinear scalar equation
\begin{equation}
f(\theta) := \error^h - \left[ \frac{\partial \qoi^h}{\partial \bs{u}^h}
\biggr|_{\bs{u}^*(\theta)} \right] \bs{e}^h = 0.
\label{eq:theta_nonlinear}
\end{equation}
We solve this problem with Newton's method, where the iterations 
\begin{equation}
\theta_{i+1} = \theta_i - \frac{f(\theta_i)}{f'(\theta_i)},
\end{equation}
are performed until $|f(\theta)| < \text{TOL}$ is satisfied. Here, $\text{TOL}$
denotes some user-specified convergence tolerance and $f'$ denotes the
derivative of the function $f$ with respect to $\theta$, which can be expressed
using the chain rule as
\begin{equation}
f'(\theta) = -\left[ \bs{e}^h \right]^T \left[ \frac{\partial^2 \qoi^h}{
\partial \bs{u}^h \partial \bs{u}^h} \biggr|_{\bs{u}^*(\theta)} \right]
\bs{e}^h.
\end{equation}
Here, the cost of solving the nonlinear scalar problem~\eqref{eq:theta_nonlinear} 
involves only assembly and does not require matrix inversion.

Unfortunately, no direct analogue to the mean-value theorem exists for
vector-valued functions~\cite{nocedal1999numerical}. Thus, the residual
linearization error $\bs{E}^{\residop}_L$ must be accounted for in a
different manner. Given the vector $\bs{u}^*$, we introduce the
\emph{modified adjoint problem}:
\begin{equation}
\left[ \frac{\partial \bs{R}^h}{\partial \bs{u}^h} \biggr|_{\bs{u}^H_h}
\right]^T \bs{z}^* = \left[ \frac{\partial \qoi^h}{\partial \bs{u}^h}
\biggr|_{\bs{u}^*} \right]^T.
\label{eq:modified_adjoint}
\end{equation}
Solving for the discretization error $\bs{e}^h$ in the residual Taylor
expansion~\eqref{eq:resid_taylor}, substituting the resulting expression
for $\bs{e}^h$ into the modified QoI Taylor expansion~\eqref{eq:qoi_taylor_exact}, 
and using the solution $\bs{z}^*$ to the modified 
adjoint~\eqref{eq:modified_adjoint} problem yields the exact error representation
\begin{equation}
\error^h = - \bs{z}^* \cdot \left( \bs{R}^h(\bs{u}^H_h) + \bs{E}^{\residop}_L
\right).
\end{equation}

For the purposes of mesh adaptivity, it is desirable to express the
error only in terms of the residual scaled by a weighting vector. To this end,
we endeavor to find a modified vector $\bs{z}^{**}$ that satisfies
\begin{equation}
\bs{z}^{**} \cdot \bs{R}^h(\bs{u}^H_h) = \bs{z}^* \cdot \left(
\bs{R}^h(\bs{u}^H_h) + \bs{E}^{\residop}_L \right).
\end{equation}
The exact least-squares solution to the above equality is
\begin{equation}
\bs{z}^{**} = \bs{z}^* + \frac{\bs{z}^* \cdot \bs{E}^{\residop}_L}{
\bs{R}^h(\bs{u}^H_h) \cdot \bs{R}^h(\bs{u}^H_h)} \bs{R}^h(\bs{u}^H_h),
\end{equation}
and the QoI discretization between the coarse and fine space can be
exactly expressed as
\begin{equation}
\error^h = \eta_2 := -\bs{z}^{**} \cdot \bs{R}^h(\bs{u}^H_h).
\label{eq:modified_adjoint_weighted_residual}
\end{equation}

Unlike the traditional adjoint-weighted residual error estimate $\eta_1$, the
estimate $\eta_2$ accounts for the contributions to the error that arise
from the residual and QoI linearization errors, $\bs{E}^{\residop}_L$ and
$\error^{\qoi}_L$, respectively. Previous techniques that include higher-order
effects in goal-oriented error estimates have been derived in a continuous adjoint
setting and include the strong form of the adjoint residual in their 
evaluation~\cite{becker2001optimal, csimcsek2015duality, endtmayer2020two,
dolejsi2023goaloriented}, a term which can be difficult to compute or even
derive~\cite{richter2015variational}. In the present discrete adjoint context,
the estimate $\eta_2$ includes higher-order terms without requiring an
evaluation of the adjoint residual.

The evaluation of $\eta_2$ is not without practical downsides, however, as it
requires the solution of the nonlinear primal problem on the fine 
space~\eqref{eq:primal_fine} to evaluate the term $\bs{e}^h$, which is present
in the nonlinear scalar equation~\eqref{eq:theta_nonlinear} and in the
evaluation of the residual linearization error~\eqref{eq:resid_taylor}.
This is a computationally expensive proposition. Similarly to
Nochetto et al.~\cite{nochetto2009safeguarded}, we propose that this method
can be used to safeguard traditional adjoint-weighted residual methods
at coarse mesh resolutions, where the solution of the fine space nonlinear
problem remains tractable. At such mesh resolutions, the traditional
adjoint-weighted residual estimate $\eta_1$ may severely under-predict
the QoI discretization error and cause an iterative adaptive algorithm to
terminate prematurely. We demonstrate the under-prediction behavior of $\eta_1$
in Section~\ref{sec:results} and have provided a potential explanation
for this behavior for nonlinear QoIs in Section~\ref{ssec:motivation}.

In fact, it is natural to question why it is even necessary to evaluate the term
$\bs{z}^{**}$ when the QoI discretization error $\error^h$ can be found exactly
when the coarse space solution $\bs{u}^H$ and fine space solution $\bs{u}^h$
are known. The answer, analogously to the traditional adjoint-weighted residual
approach, is that the term $\bs{z}^{**}$ provides \emph{local} information
about the spatial distribution of the QoI discretization error which can
be used for the purposes of mesh adaptivity. In particular, the term
$\bs{z}^{**}$ also provides local information about the residual and
QoI linearization errors $\bs{E}^{\residop}_L$ and $\error^{\qoi}_L$, 
respectively, whereas the traditional adjoint solution $\bs{z}^h$ does not.
As an avenue for future work, one could attempt to approximate the term
$\bs{e}^h$ with patch recovery techniques~\cite{zienkiewicz1992superconvergent,
wiberg1994enhanced, rodenas2007improvement, gonzalez2014mesh} to reduce
the computational cost of the newly proposed estimate.

Lastly, we remark that while the estimate $\eta_2$ exactly represents the QoI
discretization error $\error^h$ between the coarse and the fine spaces, there
is a remaining error term $\qoi(u) - \qoi^h(\bs{u}^h)$, such that the exact
QoI discretization is given as $\error = \eta_2 + \qoi(u) - \qoi^h(\bs{u}^h)$.
Naturally, this remainder term is, in general, unknowable. However, as
another avenue for future work, one could envision using traditional
functional analysis techniques found in the goal-oriented error estimation
literature to bound this remaining term and combine its effects with the
estimate $\eta_2$.

\subsection{Quadratic Quantities of Interest}
\label{ssec:estimate1_quadratic}

\begin{prop}
If $\qoi(u)$ is a quadratic functional, then the solution to the 
problem~\eqref{eq:theta_nonlinear} is $\theta = \nicefrac12$.
\end{prop}

\begin{proof}

Let $\qoi(u)$ denote a quadratic functional. We begin by equating the
expression for $\error^h$ obtained from the Taylor 
expansion~\eqref{eq:qoi_taylor_exact} with the exact Taylor expansion for a quadratic
functional expanded about the prolonged coarse solution, yielding
\begin{equation}
\left[ \frac{\partial \qoi^h}{\partial \bs{u}^h} \biggr|_{\bs{u}^*} \right]
\bs{e}^h =
\left[ \frac{\partial \qoi^h}{\partial \bs{u}^h} \biggr|_{\bs{u}^H_h} \right]
\bs{e}^h +
\frac12 \left[ \bs{e}^h \right]^T
\left[ \frac{\partial^2 \qoi^h}{\partial \bs{u}^h \partial \bs{u}^h} \right]
\bs{e}^h.
\end{equation}
Because the QoI is quadratic, its first derivative is a linear operator acting
on the solution state and its second derivative is constant, independent of
the solution state. The above equality must hold for arbitrary discretization
errors $\bs{e}^h$, with 
\begin{equation}
\left[ \frac{\partial \qoi^h}{\partial \bs{u}^h} \biggr|_{\bs{u}^*} \right]
= \left[ \frac{\partial \qoi^h}{\partial \bs{u}^h} \biggr|_{\bs{u}^H_h}
\right] +
\frac12 \left[ \bs{e}^h \right]^T
\left[ \frac{\partial^2 \qoi^h}{\partial \bs{u}^h \partial \bs{u}^h} \right].
\end{equation}
Due to the linearity of the QoI's first derivative, we can write
\begin{equation}
\left[ \frac{\partial \qoi^h}{\partial \bs{u}^h}
\biggr|_{\bs{u}^* - \bs{u}^H_h} \right]
= \frac12 \left[ \bs{e}^h \right]^T
\left[ \frac{\partial^2 \qoi^h}{\partial \bs{u}^h \partial \bs{u}^h} \right],
\end{equation}
and because the QoI is quadratic we can write
\begin{equation}
\left( \bs{u}^* - \bs{u}^H_h \right)^T
\left[ \frac{\partial^2 \qoi}{\partial \bs{u}^h \partial \bs{u}^h}
\right]
= \frac12 \left[ \bs{e}^h \right]^T
\left[ \frac{\partial^2 \qoi}{\partial \bs{u}^h \partial \bs{u}^h} \right].
\end{equation}
Using the definition of $\bs{e}^h$ and because the above expansions hold for
arbitrary quadratic QoIs, we obtain
\begin{equation}
\bs{u}^* = \frac12 \left( \bs{u}^H_h + \bs{u}^h \right).
\end{equation}
From the definition of $\bs{u}^*$, the result $\theta = \nicefrac12$ follows
immediately.
\end{proof}

This result allows us to entirely circumvent a nonlinear solve when a
quadratic QoI is chosen, and motivates an initial guess of $\theta_0 = \nicefrac12$
for Newton's method when solving~\eqref{eq:theta_nonlinear}.

\section{Adjoint Verification}
\label{sec:verification}

One surprising benefit of studying linearization errors manifests itself
in the realm of verification. For complex nonlinear physics, it is natural to
question whether the implementation of the adjoint 
problem~\eqref{eq:adjoint_fine} leads to the correct adjoint solution $\bs{z}^h$.
Answering this question with traditional verification techniques may be a
burdensome and non-trivial task. It turns out, however, that exactly solving
for the residual linearization error $\bs{E}^{\residop}_L$ provides a
convenient and straightforward mechanism for adjoint verification.

This verification procedure can be described as follows. Assume that we have
the machinery to compute every term on the right-hand side of 
equation~\eqref{eq:resid_taylor}, so that we can solve for the residual linearization
error $\bs{E}^R_L$ exactly. Additionally, assume that we have the machinery to
solve the adjoint problem~\eqref{eq:adjoint_fine} on the fine space for
$\bs{z}^h$. Choose $\qoi(u)$ to be a \emph{linear} functional, such that
$\error^{\qoi}_L = 0$.  We can exactly express the error $\error^h$ in this
context as a function of the adjoint solution $\bs{z}^h$:
\begin{equation}
\begin{aligned}
\error^h &:= \qoi^h(\bs{u}^h) - \qoi^h(\bs{u}^H_h), \\
&= -\bs{z}^h \cdot \left(\bs{R}^h(\bs{u}^H_h)+\bs{E}^{\residop}_L \right), \\
&= \eta_1 + \eta^{\residop}_L,
\end{aligned}
\label{eq:verify}
\end{equation}
where $\eta^{\residop}_L := -\bs{z}^h \cdot \bs{E}^{\residop}_L$. As a
verification check, we ensure that the effectivity $\mathcal{I}^v$, defined as
\begin{equation}
\mathcal{I}^v := \frac{\eta_1 + \eta^{\residop}_L}{\error^h}
\label{eq:verify_effectivity}
\end{equation}
evaluates to 1. Here we make several remarks. First, the evaluation of
$\bs{E}^{\residop}_L$ from the Taylor expansion~\eqref{eq:resid_taylor}
involves only data obtained from the primal problem. Thus, if one has
confidence in the correctness of the primal problem implementation, then the
linearization error is accurately obtained through linear algebra.
Second, the denominator of the effectivity~\eqref{eq:verify_effectivity}
contains only data from the primal problem and is independent of the adjoint
solution. This implies that $\mathcal{I}^v = 1$ is a true verification check
in the sense that it relates the adjoint solution to a computable value that
does not involve the adjoint implementation.

For codes that were not designed with two-level error estimation in mind but
that still have an adjoint capability, the above verification procedure can be
performed using $h$-enrichment for the fine space. In this context, the primal
problem can readily be solved on the coarse and fine mesh. The only source of
additional development effort stems from the prolongation of the coarse
solution to the fine mesh and the evaluation of the terms in 
equation~\eqref{eq:resid_taylor}. Depending on software design, the latter step may
not be trivial, but seems well worth the effort for the verification benefit it
provides.

\section{Mesh Adaptivity}
\label{sec:adapt}

\subsection{Error Localization}
\label{ssec:localization}

To control discretization errors, it is necessary to \emph{localize} the error
estimate into contributions to the error at the mesh entity level. These
contributions are often called \emph{correction indicators}. Both estimates,
$\eta_1$ and $\eta_2$, take the form of the PDE residual weighted by an
adjoint vector. Because the original variational problem~\eqref{eq:fem_form}
is posed with a Galerkin finite element method, Galerkin orthogonality
allows us to write equivalent expressions for the estimates $\eta_1$ and
$\eta_2$ as
\begin{equation}
\eta_1 = -\left( \bs{z}^h - \bs{z}^h_H \right) \cdot
\bs{R}^h(\bs{u}^H_h),
\label{eq:eta1_subtracted}
\end{equation}
and
\begin{equation}
\eta_2 = -\left( \bs{z}^{**} - \bs{z}^{**}_H \right ) \cdot
\bs{R}^h(\bs{u}^H_h),
\label{eq:eta2_subtracted}
\end{equation}
respectively. Here $\bs{z}^h_H$ and $\bs{z}^{**}_H$ denote the restriction
via interpolation of the adjoint solutions $\bs{z}^h$ and $\bs{z}^{**}$
onto the coarse space $\fspace^H$, respectively. The subtraction of the
coarse space representation of the adjoint solution is common in the
literature and is done, in part, so that the residual is weighted
by a term that indicates where the adjoint is \emph{poorly approximated}
rather than by a term that indicates where the adjoint is large in 
value~\cite{kast2017introduction}.

A common localization approach in finite volume and discontinuous Galerkin 
methods considers decomposing the adjoint-weighted residual into a sum of
absolute values of element-level adjoint-weighted residuals, for instance as
$\eta_1 \leq \sum_{e=1}^{n_{el}} \left| \left(-(\bs{z}^h - \bs{z}^h_H) \cdot
\bs{R}^h(\bs{u}^H_h) \right) \bigr|_{\Omega_e} \right|$, where the absolute
value term at each element is used as the correction indicator. This approach,
however, may not be optimal for continuous Galerkin finite element methods as it
does not account for inter-element cancellation, where the sum of the resultant
indicators may vastly overestimate the error and, in turn, lead to a suboptimal
adaptive strategy~\cite{fidkowski2011review}.

Presently, we make use of a variational technique for localization developed
by Richter and Wick~\cite{richter2015variational}, which expresses the
estimates~\eqref{eq:eta1_subtracted} and \eqref{eq:eta2_subtracted}
analogously in variational form and achieves localization by inserting a
partition of unity $\sum_i \phi^i = 1$ into the weighting function slot of
the variational residual $\residop$, to be defined. We take this partition of
unity to be linear Lagrange basis functions which results in the error being
localized to the $n_{vtx}$ mesh vertices. Concretely, this results in
localized error contributions $\eta^i_1$ and $\eta^i_2$, $i=1,2,\dots,n_{vtx}$
at mesh vertices of the form
\begin{equation}
\eta_1 = \sum_{i=1}^{n_{vtx}}
\underbrace{-\mathcal{R}(u^H; (z^h - z^h_H) \phi^i)}_{\eta^i_1}
\label{eq:eta1_localized}
\end{equation}
and
\begin{equation}
\eta_2 = \sum_{i=1}^{n_{vtx}}
\underbrace{-\mathcal{R}(u^H; (z^{**} - z^{**}_H) \phi^i)}_{\eta^i_2},
\label{eq:eta2_localized}
\end{equation}
where $\mathcal{R}: \fspace \times \fspace \to \reals$ denotes the variational
residual $\mathcal{R}(u,w) := \rhsop(w) - \lhsop(u;w)$ and $u^H, z^h, z^h_H,
z^{**}, z^{**}_H$ denote the function counterparts of the nodal coefficient
vectors $\bs{u}^H, \bs{z}^h, \bs{z}^h_H, \bs{z}^{**}$, and $\bs{z}^{**}_H$,
respectively, as determined by the function spaces $\fspace^H$ and $\fspace^h$.
As a final step, to compute an element level correction indicator $\eta^e$,
we interpolate the value of the vertex-based error contribution $\eta^i$ to
element centers and take its absolute value.

Here we remark that a primary motivation in the localization of the estimate
$\eta_2$ is to additionally localize the effects of the linearization errors
$\error^{\qoi}_L$ and $\bs{E}^{\residop}_L$ during the mesh adaptation process.
That is, the localized indicator $\eta^i_2$ will also indicate regions
in the domain where the linearization errors are large. By accurately
resolving the mesh in such regions, we endeavor to obtain more optimal meshes
than those obtained from the use of the traditional adjoint-weighted residual
estimate $\eta_1$.

\subsection{Mesh Size Field}
\label{ssec:size_field}

Once element-level correction indicators $\eta^e$, $e=1,2,\dots,n_{el}$
have been computed, they must be used to modify the finite element mesh.
Naturally, we would like to refine regions of the mesh that the indicators
$\eta^e$ suggest contribute strongly to the QoI discretization error and, if
applicable, coarsen regions of the mesh that do not greatly affect the QoI
discretization error. To this end, we use isotropic \emph{conformal} mesh
adaptation driven by a \emph{mesh size field}, which produces an
adapted mesh without hanging nodes. This process achieves refinement and
coarsening via a series of edge splits, swaps, and collapses~\cite{li20053d,
alauzet2006parallel}. The mesh size field indicates the desired characteristic
element size in the adapted mesh. We choose a mesh size field that attempts
to equidistribute the error in the resultant mesh given a target $T$ number of
elements, as described by Bousetta et al.~\cite{boussetta2006adaptive}.
Here, the size field is found by assuming that the error is uniformly
distributed in the resultant adapted mesh and comparing the ratio
of the element contribution to the error in the current mesh to the
desired adapted mesh as related by expected convergence rates. This leads to an
optimization problem that can be solved analytically for the new mesh size.
Let $p$ denote the polynomial interpolant order of the function space
$\fspace^H$ in $d$ spatial dimensions. Let $|\Omega^e_H| :=
\text{meas}(\Omega^e)$ denote the characteristic mesh
size of element $\Omega^e$. The desired new element size $|\Omega^e_H|^{
\text{new}}$ can then be expressed as
\begin{equation}
\left| \Omega^e_H \right|^{\text{new}} = \left( T^{-1} \sum_{e=1}^{n_{el}}
\left( \eta^e \right)^{\frac{2d}{2p + d}} \right)^{\frac{1}{d}}
\left( \eta^e \right)^{\frac{-2}{2p + d}} \left| \Omega^e_H \right|.
\label{eq:size_field}
\end{equation}
As a final consideration we clamp the value of the new element size to be
no less than half of the current element size and no greater than twice
the current element size, such that $\nicefrac12 |\Omega^e_H| \leq
|\Omega^e_H|^{\text{new}} \leq 2 | \Omega^e_H|$.

\section{Results}
\label{sec:results}

\subsection{A Nonlinear Poisson's Problem}
\label{ssec:model_problem}

As a model problem, we consider a nonlinear Poisson's equation, written as
\begin{equation}
\begin{cases}
\begin{aligned}
-\nabla \cdot [(1+\alpha u^2) \nabla u] &= f,
&& \text{in} \; \Omega, \\
u &= 0, && \text{on} \; \Gamma. \\
\end{aligned}
\end{cases}
\label{eq:strong_form}
\end{equation}
where $\alpha \in \reals_{\geq 0}$ denotes a scaling parameter that controls
the degree of nonlinearity. The forcing function $f$ is assumed
to be sufficiently smooth, such that $f \in L^2(\Omega)$.
Let $\fspace$ denote the Sobolev space $H^1(\Omega)$ whose trace vanishes on
the boundary $\Gamma$, such that $\fspace := \left\{ v \in H^1(\Omega) :
v = 0 \; \text{on} \; \Gamma \right\}$. Placing the model problem in weak
form yields the semilinear form
\begin{equation}
\lhsop(u; w) := \int_{\Omega} (1 + \alpha u^2) \nabla u \cdot \nabla w \,
\text{d} \Omega,
\label{eq:lhs}
\end{equation}
and the linear functional
\begin{equation}
\rhsop(w) := \int_{\Omega} f w \, \text{d}\Omega.
\end{equation}
Let $\fspace^H$ and $\fspace^h$ denote finite dimensional function spaces
defined as
\begin{equation}
\fspace^H := \left\{ u^H : u^H \in \fspace, \, u^H |_{\Omega^e} \in
\mathbb{P}^1(\Omega)^{\dims} \right\},
\end{equation}
and
\begin{equation}
\fspace^h := \left\{ u^h : u^h \in \fspace, \, u^h |_{\Omega^e} \in
\mathbb{P}^2(\Omega)^{\dims} \right\},
\end{equation}
respectively, where $\mathbb{P}^p$ denotes the space of piecewise polynomials
of order $p$ over the elements $\Omega^e$, $e=1,2,\dots,n_{el}$.

For all results below, we consider a square geometry with a square hole,
such that $\Omega := \{ \bs{x} : \bs{x} \in (-1,1) \times (-1,1) \setminus 
[-\nicefrac12,\nicefrac12] \times [-\nicefrac12, \nicefrac12] \}$,
discretized with triangular elements. To eliminate or drastically reduce
additional sources of numerical errors, all results use the following
operations: requisite integral evaluations are computed using a $6^{th}$
order, 12 point Gaussian quadrature rule over triangles, nonlinear
systems are solved using an absolute tolerance of $\text{TOL}=10^{-10}$,
linear systems are solved with GMRES using an absolute tolerance
of $10^{-12}$, and derivative quantities (the residual Jacobian, the residual
Jacobian adjoint, the QoI gradient, and the QoI Hessian) are computed with
automatic differentiation.

We consider the following four quantities of interest:
\begin{equation}
\begin{aligned}
\qoi_1(u) &:= \int_{\Omega} u \, \text{d} \Omega, \\
\qoi_2(u) &:= \int_{\Omega_s} u^3 \, \text{d} \Omega, \\
\qoi_3(u) &:= \int_{\Omega_s} \nabla u \cdot \nabla u \, \text{d} \Omega, \\
\qoi_4(u) &:= \int_{\Omega_s} \sqrt{\nabla u \cdot \nabla u} \, \text{d} \Omega.
\end{aligned}
\end{equation}
Here, $\Omega_s := \{ \bs{x} : \bs{x} \in (0,1) \times (-1,0) \setminus
[0,\nicefrac12] \times [-\nicefrac12,0] \}$ denotes a subdomain of the overall
domain $\Omega$. Figure~\ref{fig:results_domain_and_mesh} illustrates the
domain $\Omega$, its initial mesh, and the definition of the subdomain
$\Omega_s$.
 
\begin{figure}[ht]
\centering
\includegraphics[width=0.4\textwidth]{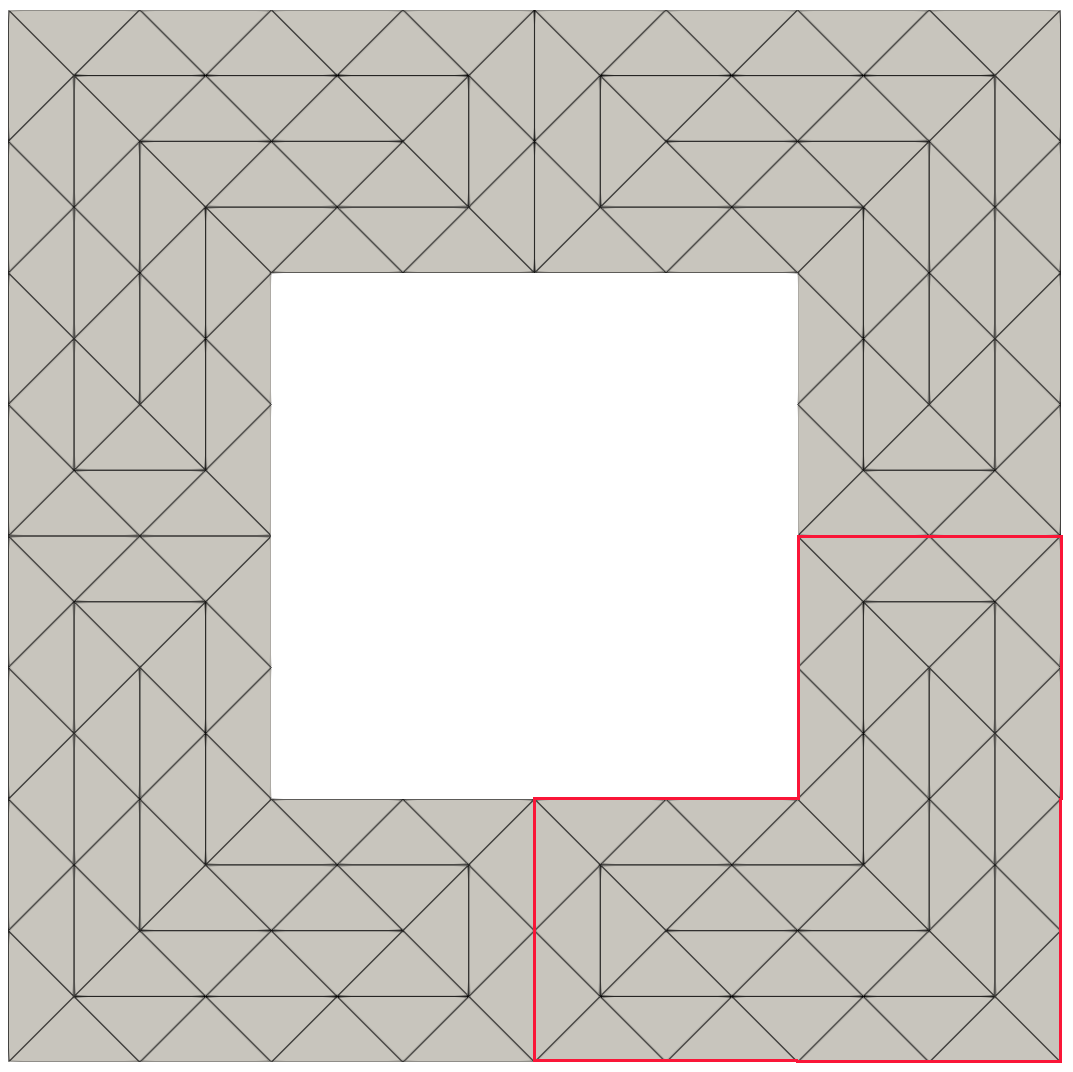}
\caption{The domain $\Omega$, its initial mesh, and the subdomain $\Omega_s$
(red outline).}
\label{fig:results_domain_and_mesh}
\end{figure}

\subsubsection{A Manufactured Solution}
\label{ssec:manufactured}

We first consider a manufactured solution on the domain $\Omega$ of the form:
\begin{equation}
u(x,y) = \sin(2 \pi x) \sin(2 \pi y) \exp(\nicefrac52 (x+y)),
\end{equation}
which satisfies the homogeneous Dirichlet boundary conditions.
The right-hand side forcing function $f$ is found by 
substituting the exact expression for $u$ into the original governing
equations~\eqref{eq:strong_form}. The solution and its gradient components are
shown in Figure~\ref{fig:manufactured_soln}. The numerical values of the
QoIs $\qoi_i, \; i=1,2,3,4,$ are given in~\ref{sec:qoi_values_manufactured}.

\begin{figure}[ht]
\centering
\begin{subfigure}{0.33\textwidth}
\centering
\includegraphics[width=.99\linewidth]{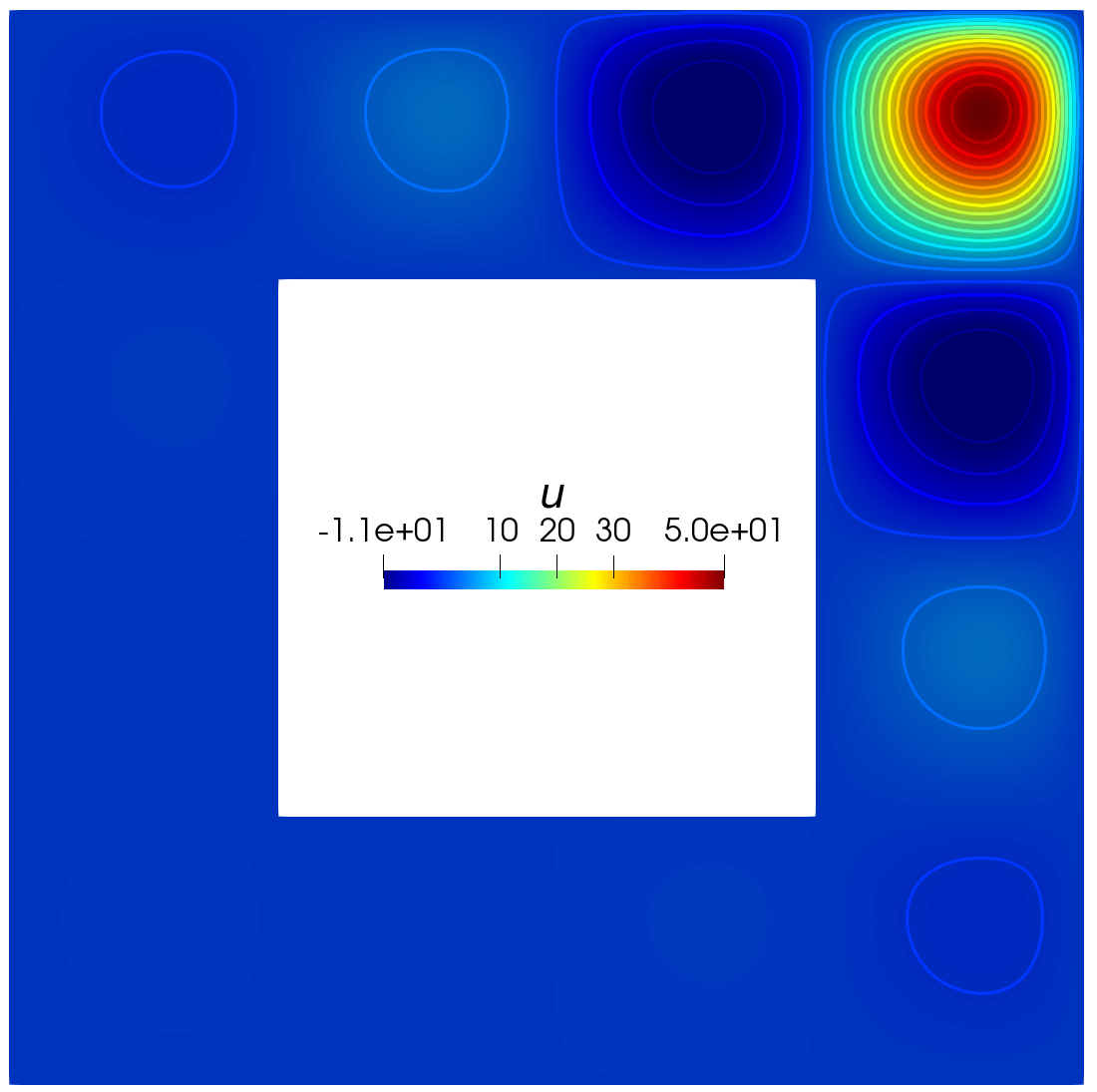}
\end{subfigure}%
\begin{subfigure}{0.33\textwidth}
\centering
\includegraphics[width=.99\linewidth]{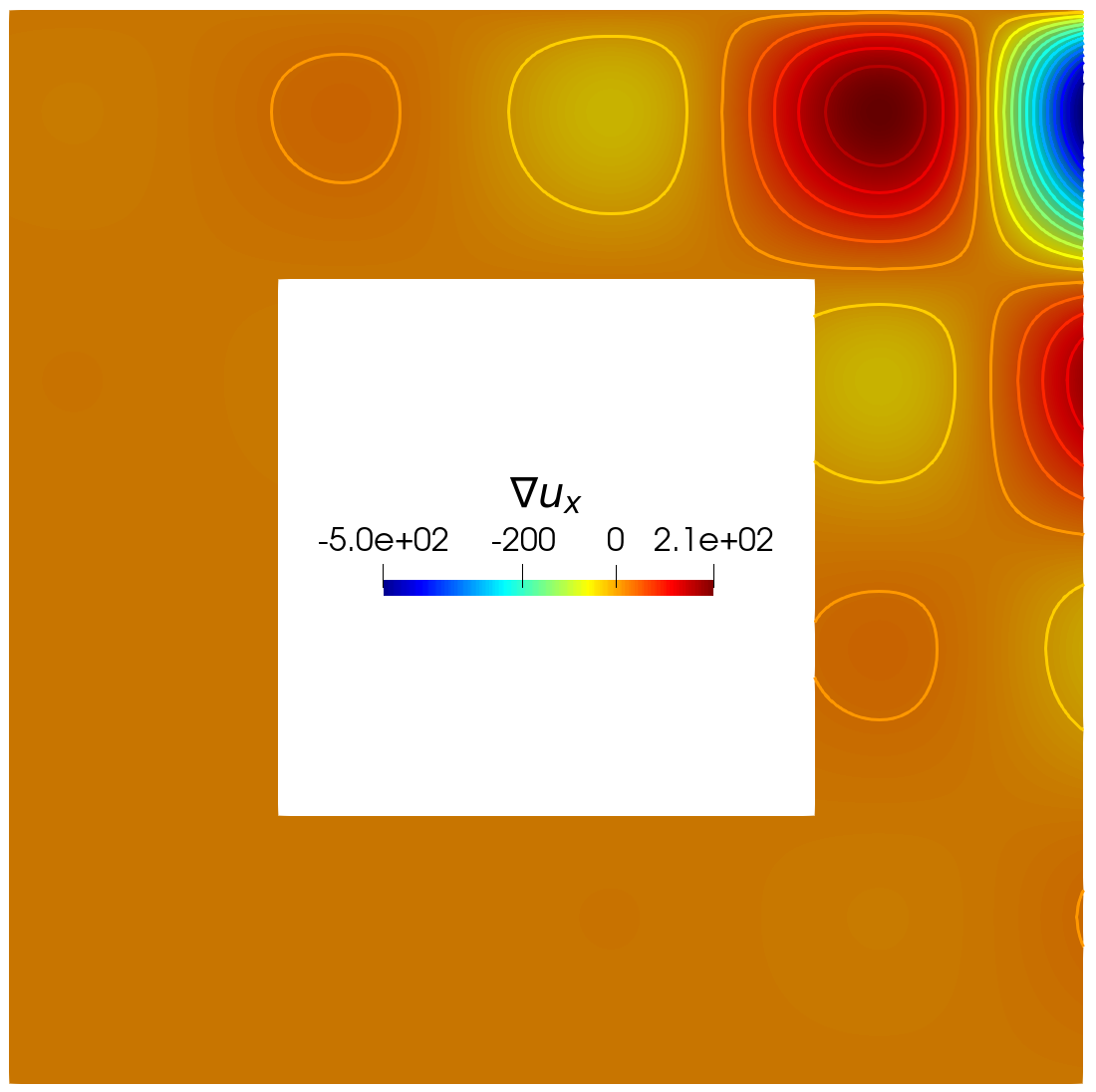}
\end{subfigure}%
\begin{subfigure}{0.33\textwidth}
\centering
\includegraphics[width=.99\linewidth]{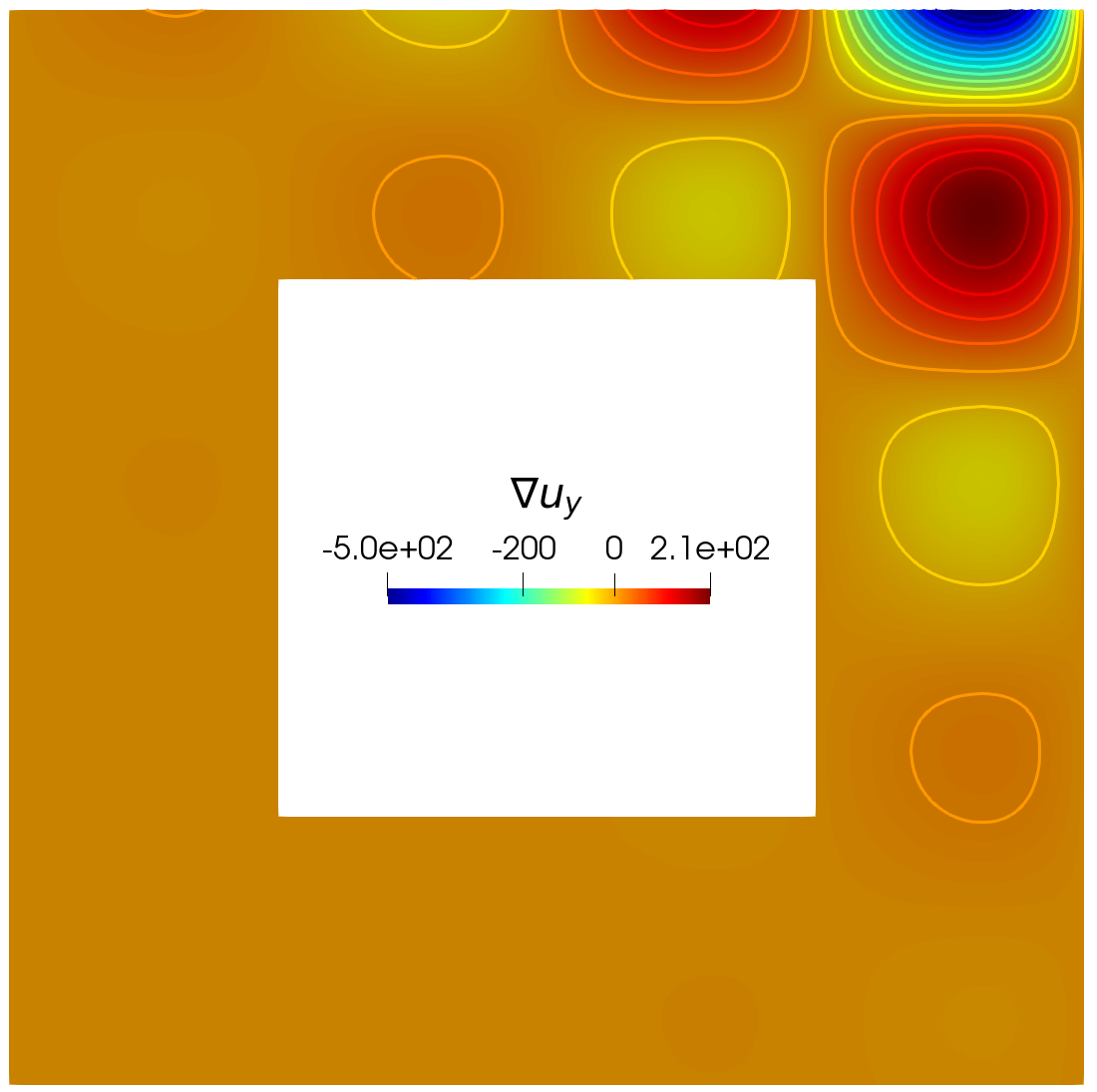}
\end{subfigure}
\caption{The manufactured solution (left), the $x$-component of the
manufactured solution gradient (center), and the $y$-component of the
manufactured solution gradient (right) with $20$ linearly-spaced contour
lines.}
\label{fig:manufactured_soln}
\end{figure}

\subsubsection{Adjoint Verification}

As an initial test with the manufactured solution, we consider the mesh with
$192$ elements shown in Figure~\ref{fig:results_domain_and_mesh}. We then
consider the following values of $\alpha \in \left\{ 0, 10^{-4}, 10^{-3}, 10^{-2},
10^{-1} \right\}$, where $\alpha=0$ corresponds to a linear PDE and increasing
values of $\alpha$ correspond to a greater degree of nonlinearity in the
PDE. We solve the primal problem on the coarse and fine spaces, $\fspace^H$ and
$\fspace^h$. For context, Newton's method takes $0,4,7,11,17$ iterations
to solve the primal problem on both spaces for the corresponding increasing
values of $\alpha$. Using the coarse and fine space primal solutions,
we compute the quantities $\error^h$,
$\eta_1$ and $\eta^{\residop}_L$, as required for the verification procedure
outlined in Section~\ref{sec:verification}. We consider the linear quantity of
interest $\qoi_1(u)$ and demonstrate that the adjoint solution $\bs{z}^h$
passes the verification check, as shown in the right-most column of
Table~\ref{tab:qoi1_verify}. Additionally, the second column of this table
illustrates that the norm of the residual linearization error
$\bs{E}^{\residop}_L$ is increasing with increasing values of $\alpha$.
Lastly, the sixth column of these tables illustrates that the traditional error
estimate $\eta_1$ does not entirely capture the QoI discretization error
$\error^h$ between the fine and coarse spaces due to the residual linearization
error $\bs{E}^{\residop}_L$. However, correcting the traditional estimate with
the term $\eta^{\residop}_L$ that corresponds to residual linearization errors
entirely recovers the error $\error^h$.

\begin{table}[ht!]
\begin{center}
\begin{tabular}{c || c | c | c | c | c | c}
$\alpha$ & $||\bs{E}^{\residop}_L||_2 $ & $\error^h$ & $\eta_1$ & $\eta^{\residop}_L$ &
$\nicefrac{\eta_1}{\error^h}$ &
$\nicefrac{(\eta_1+\eta^{\residop}_L)}{\error^h}$ \\ \hline \hline
$0$       & 1.8785e$-$13 & 3.2696e$-$01 & 3.2696e$-$01 &  1.6688e$-$15 & 1.0000e+00 & 1.0000e+00 \\ \hline 
$10^{-4}$ & 1.7321e+00 & 2.9114e$-$01 & 3.0420e$-$01 & -1.3061e$-$02 & 1.0449e+00 & 1.0000e+00 \\ \hline
$10^{-3}$ & 1.2996e+01 & 1.3244e$-$01 & 1.9139e$-$01 & -5.8949e$-$02 & 1.4451e+00 & 1.0000e+00 \\ \hline
$10^{-2}$ & 8.5432e+01 & 7.6441e$-$04 & 1.1158e$-$01 & -1.1082e$-$01 & 1.4597e+02 & 1.0000e+00 \\ \hline
$10^{-1}$ & 7.7580e+02 & 9.4845e$-$02 & 2.6463e$-$01 & -1.6978e$-$01 & 2.7901e+00 & 1.0000e+00
\end{tabular}
\end{center}
\caption{Adjoint verification data for the linear QoI $\qoi_1(u)$ for the
manufactured solution.}
\label{tab:qoi1_verify}
\end{table}

\subsubsection{Modified Adjoint Verification}
\label{ssec:second_adjoint_verification}

As a second test with the manufactured solution, we consider the mesh with
$192$ elements shown in Figure~\ref{fig:results_domain_and_mesh}, choose
$\alpha=10^{-2}$, and consider the nonlinear QoIs $\qoi_2(u), \qoi_3(u),$ and
$\qoi_4(u)$, so the residual and QoI linearization errors,
$\bs{E}^{\residop}_L$ and $\error^{\qoi}_L$, respectively, are nonzero. As a
verification check, we ensure the estimate $\eta_2$ exactly recovers the QoI
discretization error $\error^h$ between the coarse and fine space, such that
$\nicefrac{\eta_2}{\error^h} = 1$. This provides assurance that the value
$\bs{z}^{**}$ is computed correctly. Table~\ref{tab:qoi_nonlinear_verify} shows
results obtained from this verification check. In particular, the right-most
column illustrates that the traditional error estimate $\eta_1$ does not exactly
recover the error $\error^h$ in the presence of residual and linearization
errors. In contrast, the right-most column of the table demonstrates that the
newly proposed estimate $\eta_2$ does indeed exactly recover the error
$\error^h$, thus completing the verification check. Figure~\ref{fig:modified_adjoint_verify}
illustrates the adjoint solutions
$z := \lim_{H \to 0} \bs{z}^H$ for the QoIs $\qoi_2(u), \qoi_3(u)$ and
$\qoi_4(u)$. Lastly, Figure~\ref{fig:qoi2_f_theta} illustrates the form of the
function $f(\theta)$, as described by equation~\eqref{eq:theta_nonlinear}, for
the QoI $\qoi_2(u)$. Here we remark that even though the functional form of
$f(\theta)$ is quadratic in $\theta$, it remains very close to linear
behavior in the domain $\theta \in [0,1]$, and its endpoints are very nearly
equal and opposite. As a result, our initial guess of $\theta_0 = \nicefrac12$
for Newton's method seems well justified in this instance.

\begin{table}[ht!]
\begin{center}
\begin{tabular}{c || c | c | c | c | c }
 & $\error^h$ & $\eta_1$ & $\eta_2$ & $\nicefrac{\eta_1}{\error^h}$ &
$\nicefrac{\eta_2}{\error^h}$ \\ \hline \hline
$\qoi_2(u)$ & 1.1417e+00 &  8.7186e$-$01 & 1.1417e+00 &  7.6362e$-$01 & 1.0000e+00 \\ \hline
$\qoi_3(u)$ & 1.5043e+01 & -1.8944e+00 & 1.5043e+01 & -1.2593e$-$01 & 1.0000e+00 \\ \hline
$\qoi_4(u)$ & 3.3263e$-$01 & -9.4750e$-$02 & 3.3263e$-$01 & -2.8485e$-$01 & 1.0000e+00
\end{tabular}
\end{center}
\caption{Modified adjoint verification data for the nonlinear QoIs
$\qoi_2(u), \qoi_3(u), \qoi_4(u)$ for the manufactured solution
when $\alpha=10^{-2}$.}
\label{tab:qoi_nonlinear_verify}
\end{table}

\begin{figure}[ht!]
\centering
\begin{subfigure}{0.33\textwidth}
\centering
\includegraphics[width=.99\linewidth]{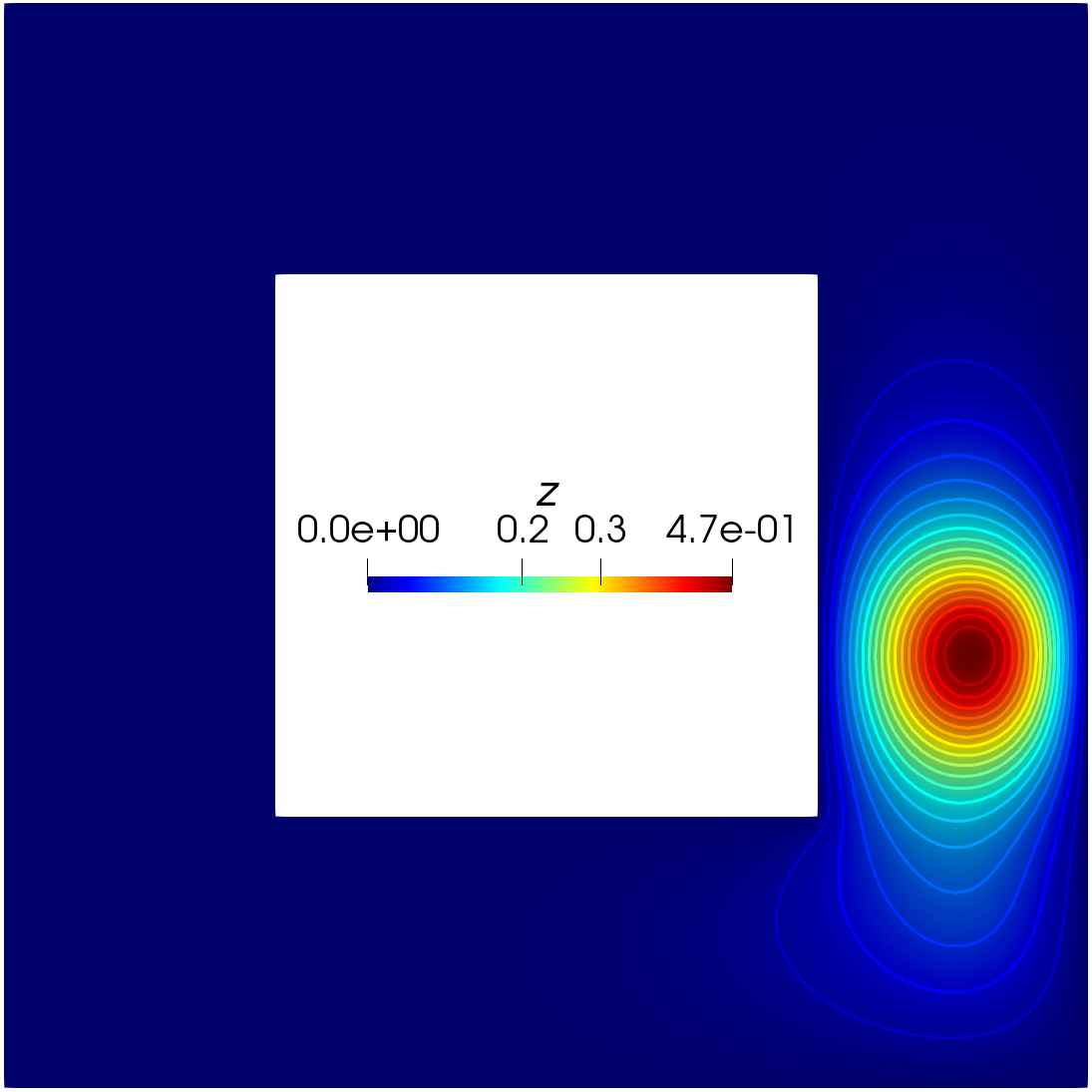}
\end{subfigure}%
\begin{subfigure}{0.33\textwidth}
\centering
\includegraphics[width=.99\linewidth]{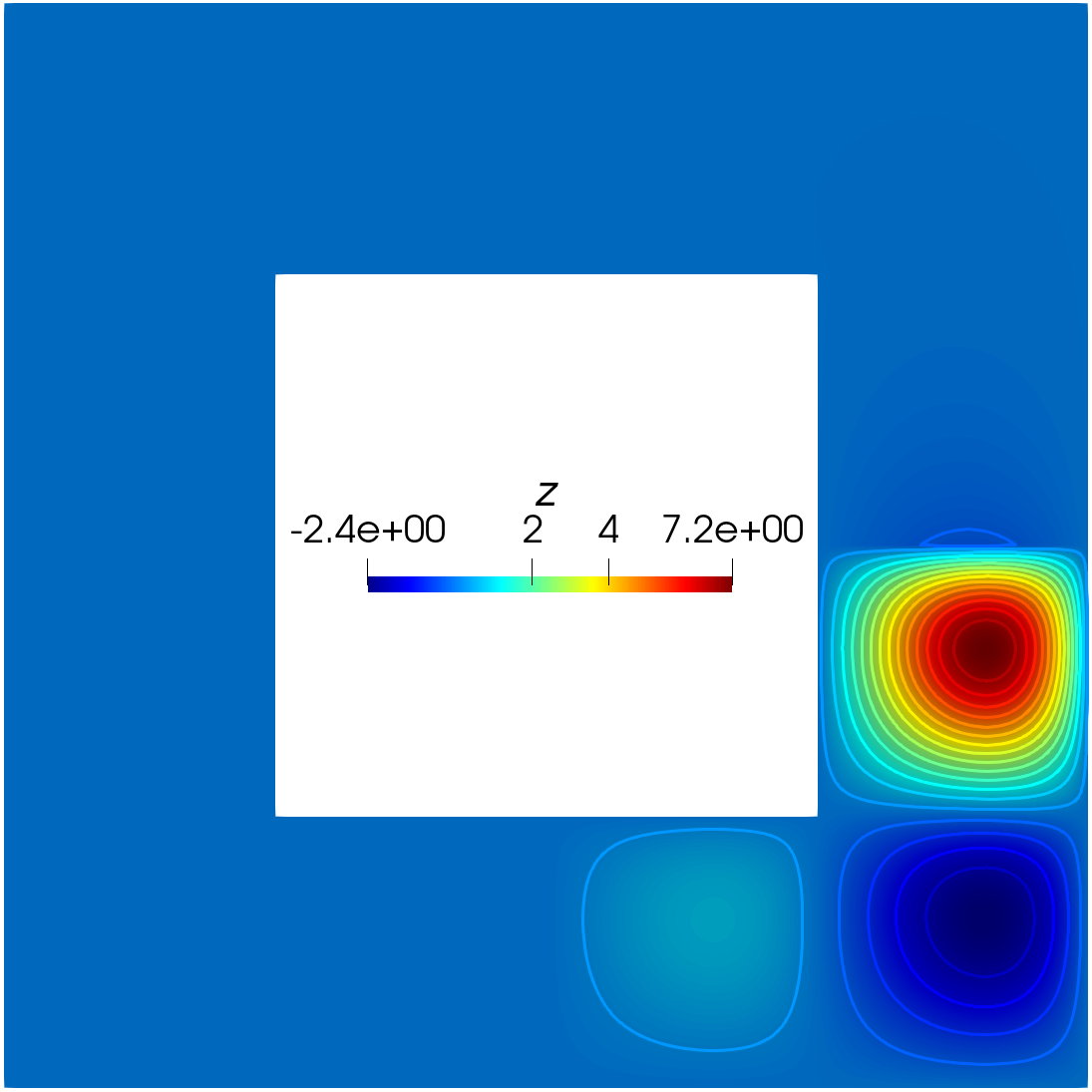}
\end{subfigure}%
\begin{subfigure}{0.33\textwidth}
\centering
\includegraphics[width=.99\linewidth]{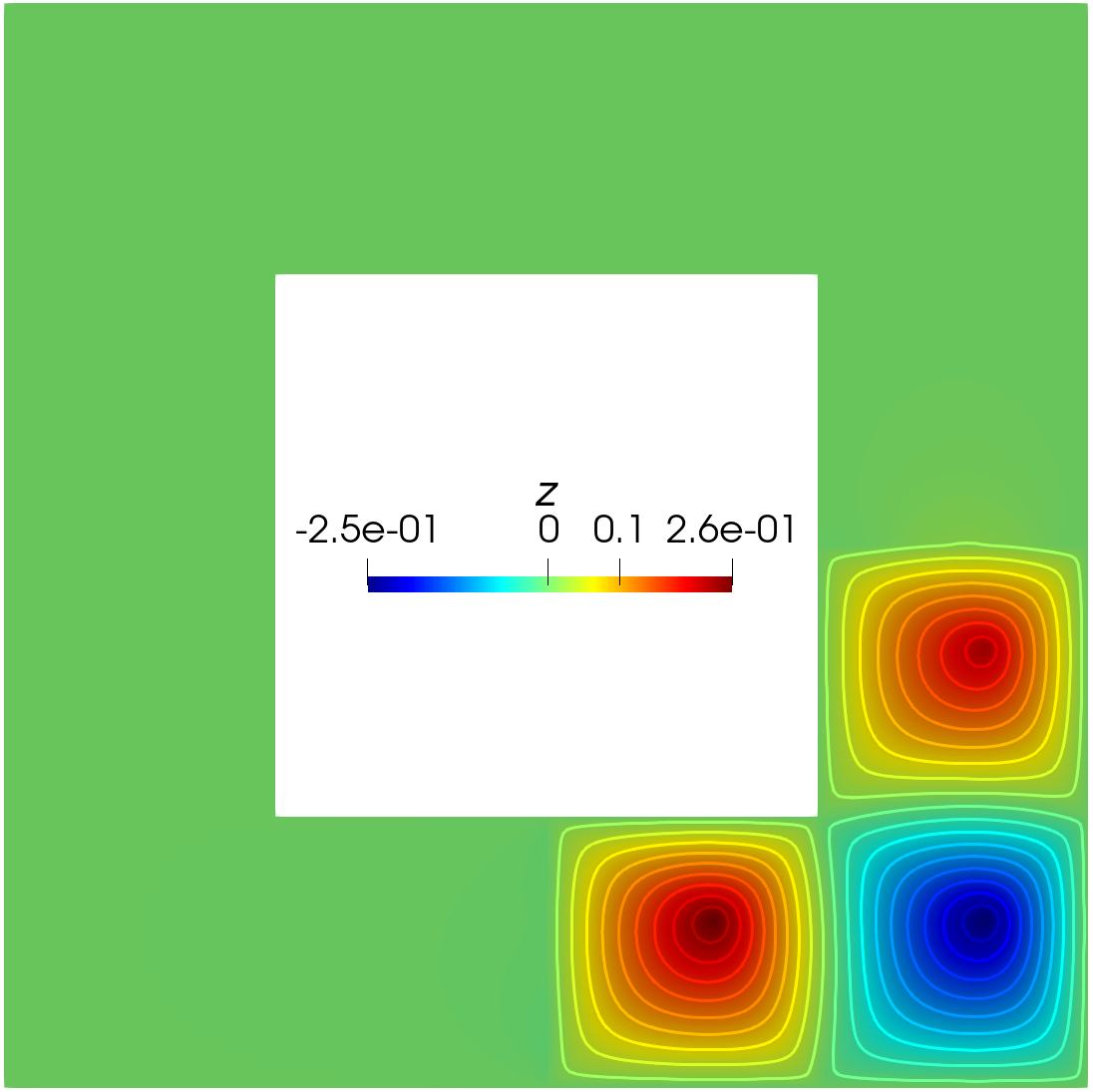}
\end{subfigure}
\caption{The adjoint solution $z$ corresponding to the model problem with a
manufactured solution when $\alpha=10^{-2}$ for the nonlinear QoIs $\qoi_2(u)$
(left), $\qoi_3(u)$ (center), and the $\qoi_4(u)$ (right) with 20 linearly
spaced contour curves.}
\label{fig:modified_adjoint_verify}
\end{figure}

\begin{figure}[ht!]
\centering
\includegraphics[width=.50\linewidth]{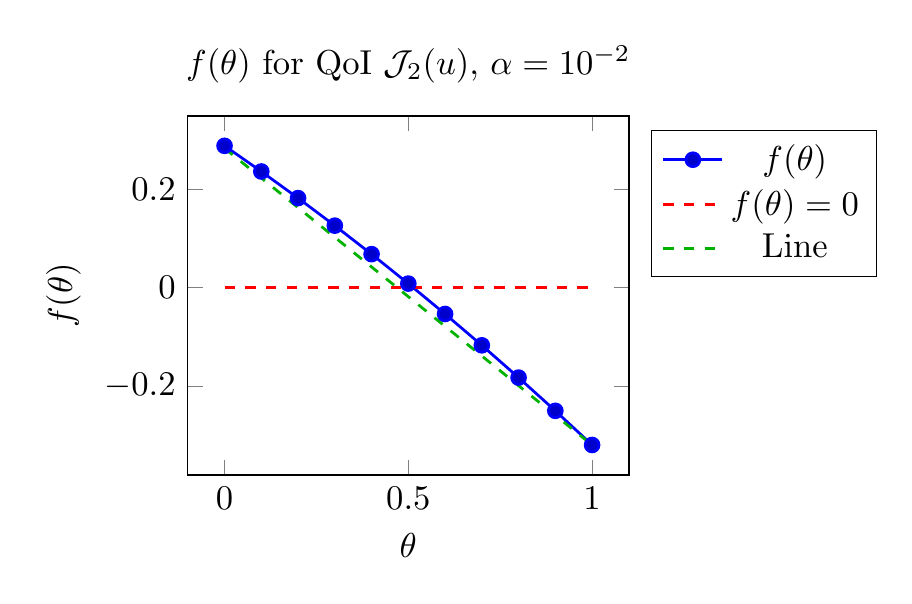}
\caption{The function $f(\theta)$ plotted at $10$ evenly spaced points in
$[0,1]$ for the QoI $\qoi_2(u)$ for the manufactured solution when
$\alpha=10^{-2}$.}
\label{fig:qoi2_f_theta}
\end{figure}

\subsubsection{Asymptotic Behavior}
\label{ssec:asymptotic_behavior}

In this section, we investigate the behavior of the nonlinear functional
quantities and the estimates $\eta_1$ and $\eta_2$ as $| \Omega^e_H | \to 0$.
For all examples, we begin with an initial mesh with $192$ elements, as shown
in Figure~\ref{fig:results_domain_and_mesh} and consider a sequence of
uniformly refined meshes, where the number of elements is $n_{el} =
192,768,3072,12288$ on this sequence of meshes. For each QoI, at each mesh
instance we solve the primal problems~\eqref{eq:primal_coarse} and
\eqref{eq:primal_fine}, solve the traditional adjoint 
equation~\eqref{eq:adjoint_fine}, solve the modified adjoint 
equation~\eqref{eq:modified_adjoint}, compute the error estimate $\eta_1$ 
from~\eqref{eq:adjoint_weighted_residual}, and compute the error estimate $\eta_2$
from~\eqref{eq:modified_adjoint_weighted_residual}, where the scaling parameter
$\alpha$ is chosen to be $\alpha = 10^{-2}$.

Using the error estimates $\eta_1$ and $\eta_2$, we compute four
effectivity indices: $\nicefrac{\eta_1}{\error}, \nicefrac{\eta_2}{\error},
\nicefrac{\eta_1}{\error^h}$, and $\nicefrac{\eta_2}{\error^h}$. The first
two effectivities compare the error estimates to the exact QoI discretization
error $\error$, while the second two effectivities compare the error
estimates to the QoI discretization error between the coarse and fine
spaces, $\error^h$. As discussed in the previous section, the effectivity
$\nicefrac{\eta_2}{\error^h}$ should evaluate identically to $1$. In contrast,
the effectivity $\nicefrac{\eta_1}{\error^h}$ illustrates the effect that
neglecting the linearization errors $\bs{E}^{\residop}_L$ and $\error^{\qoi}_L$
has on recovering the two-space QoI error $\error^h$. If these linearization
errors become small as $|\Omega^e_H| \to 0$, as one would hope, then the
effectivity $\nicefrac{\eta_1}{\error^h}$ should approach $1$. Lastly, if the
error estimates are \emph{effective}, then the effectivity indices with respect
to the exact error $\error$ will tend towards $1$ as
$|\Omega^e_H| \to 0$.

In addition to computing the effectivity indices described above, we also
compute \emph{corrected} functional approximations using the error estimates.
These corrected functional values take the form $\qoi^H(\bs{u}^H) + \eta_1$
for the traditional error estimate $\eta_1$ and $\qoi^H(\bs{u}^H) + \eta_2$ 
for the newly proposed estimate $\eta_2$. The corrected functional values
should provide a more accurate representation of the QoI, which we measure by
computing their errors with respect to the exact solution $\qoi(u) -
(\qoi^H(\bs{u}^H) + \eta_1)$ and $\qoi(u) - (\qoi^H(\bs{u}^H) + \eta_2)$. If
the correction improves the accuracy of the QoI approximation, then these
computed errors will converge at a faster rate than the exact QoI
discretization error, $\qoi(u) - \qoi^H(\bs{u}^H)$. 
Figures~\ref{fig:qoi2_manufactured_plots}, \ref{fig:qoi3_manufactured_plots},
and \ref{fig:qoi4_manufactured_plots} illustrate the asymptotic behaviors of
the error estimates $\eta_1$ and $\eta_2$ for the QoIs $\qoi_2(u), \qoi_3(u)$,
and $\qoi_4(u)$, respectively. For each figure, the plot on the left 
illustrates the behavior of the effectivity indices and the plot on the right
illustrates the behavior of computed error metrics as $|\Omega^e_H| \to 0$.

\begin{figure}[ht!]
\centering
\begin{subfigure}{0.52\textwidth}
\centering
\includegraphics[width=.99\linewidth]{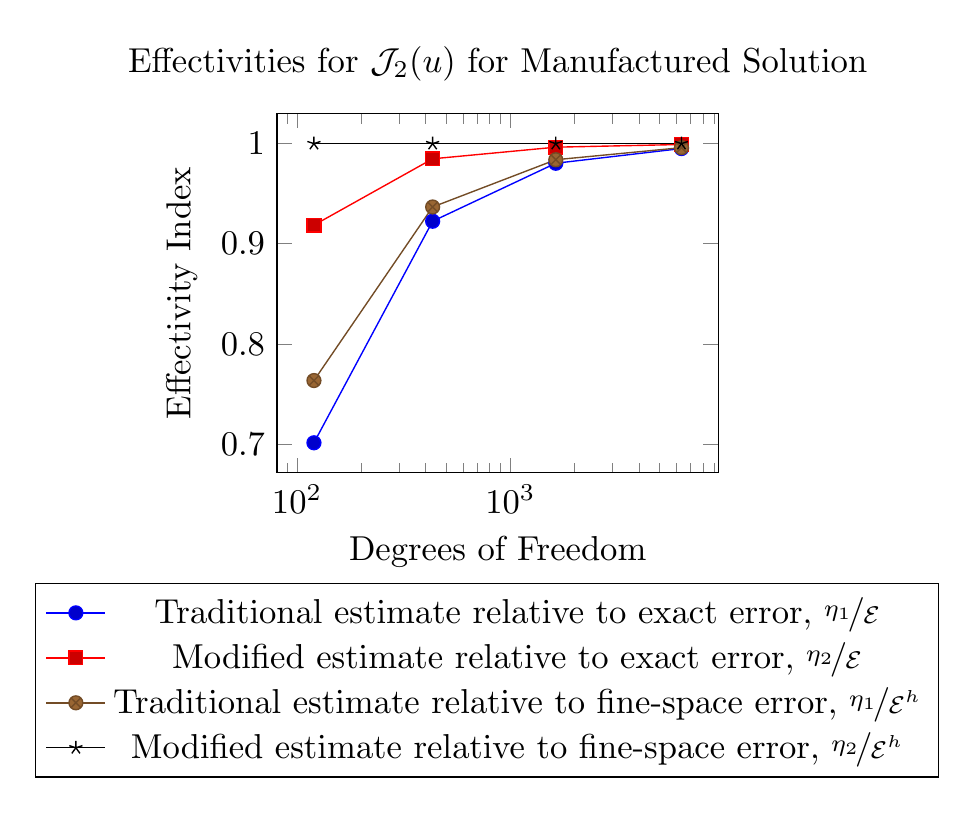}
\end{subfigure}%
\begin{subfigure}{0.47\textwidth}
\centering
\includegraphics[width=.99\linewidth]{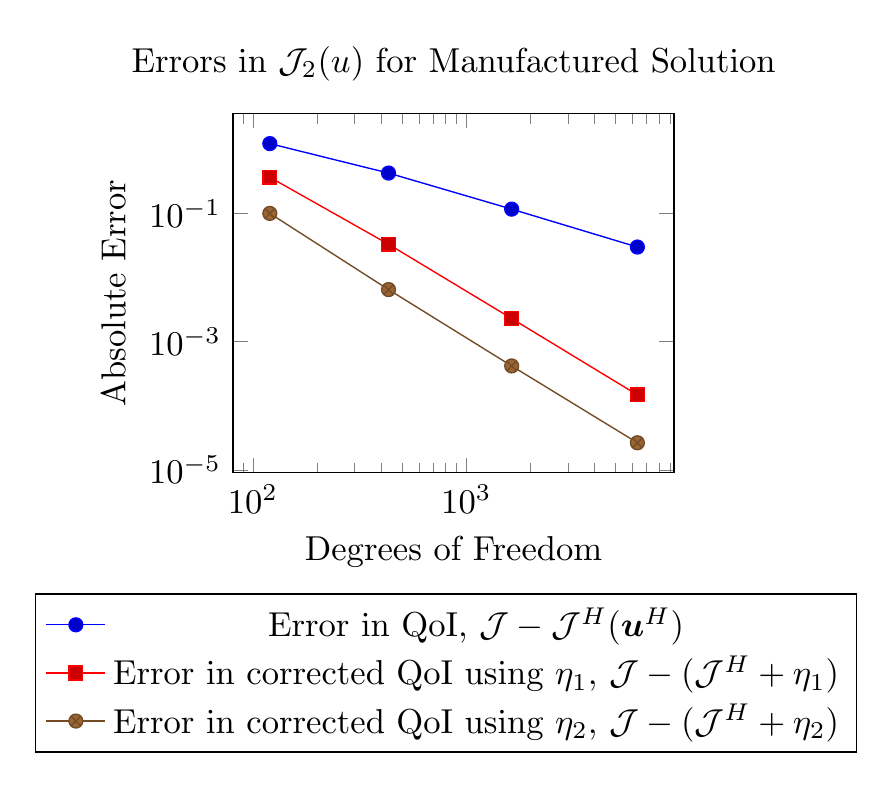}
\end{subfigure}
\caption{Asymptotic behavior of the estimates $\eta_1$ and $\eta_2$
for the QoI $\qoi_2(u)$ for the manufactured solution when $\alpha=10^{-2}$
on a sequence of uniformly refined meshes.}
\label{fig:qoi2_manufactured_plots}
\end{figure}

Figure~\ref{fig:qoi2_manufactured_plots} illustrates that both estimates
$\eta_1$ and $\eta_2$ are effective as the mesh size goes to zero for the QoI
$\qoi_2(u)$ for the chosen problem of interest, with
$\nicefrac{\eta_1}{\error} \to 1$ and $\nicefrac{\eta_2}{\error} \to 1$
as $|\Omega^e_H| \to 0$. However, the plot on the left of
Figure~\ref{fig:qoi2_manufactured_plots} illustrates that the newly proposed
estimate $\eta_2$ is more effective than that of the traditional estimate
$\eta_1$ at coarse mesh resolutions. Additionally, the plot on the left 
of Figure~\ref{fig:qoi2_manufactured_plots} demonstrates that both estimates
result in a corrected functional evaluation that converges at a faster
rate than that of $\qoi_2^H(\bs{U}^H)$. In fact, the rate of this corrected 
functional is nearly identical for the two estimates, but the newly proposed
estimate $\eta_2$ results in a constant shift towards a lower error and,
as a result, a more accurate approximation of the true functional value
$\qoi_2(u)$.

\begin{figure}[ht!]
\centering
\begin{subfigure}{0.52\textwidth}
\centering
\includegraphics[width=.99\linewidth]{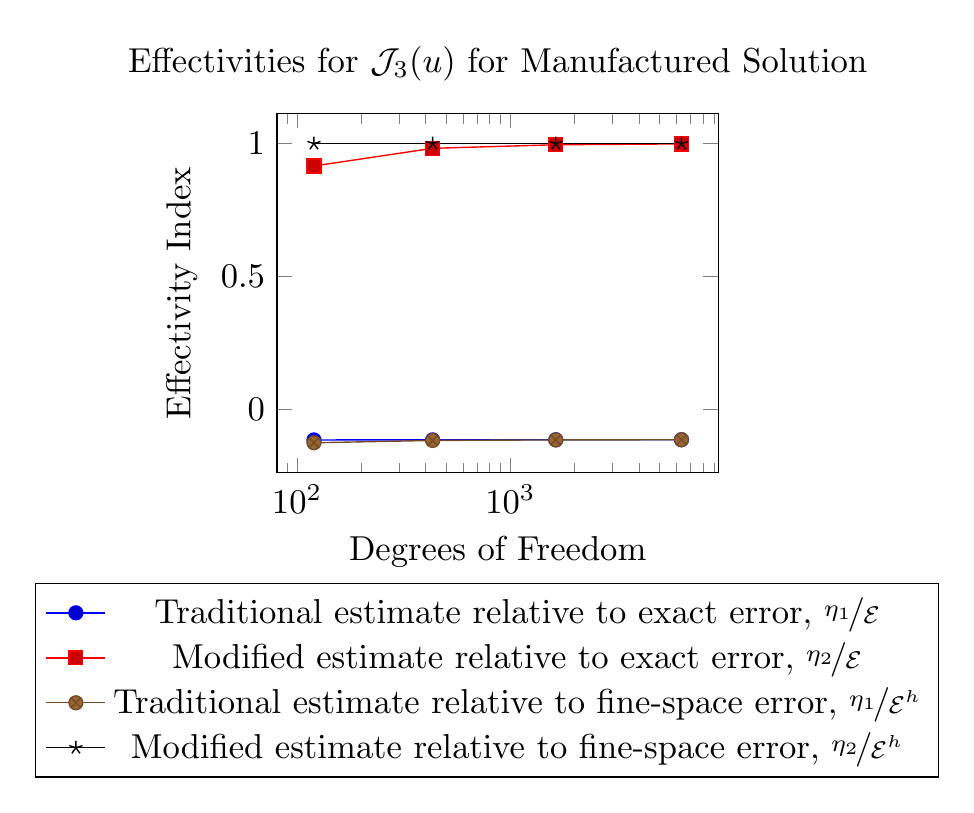}
\end{subfigure}%
\begin{subfigure}{0.47\textwidth}
\centering
\includegraphics[width=.99\linewidth]{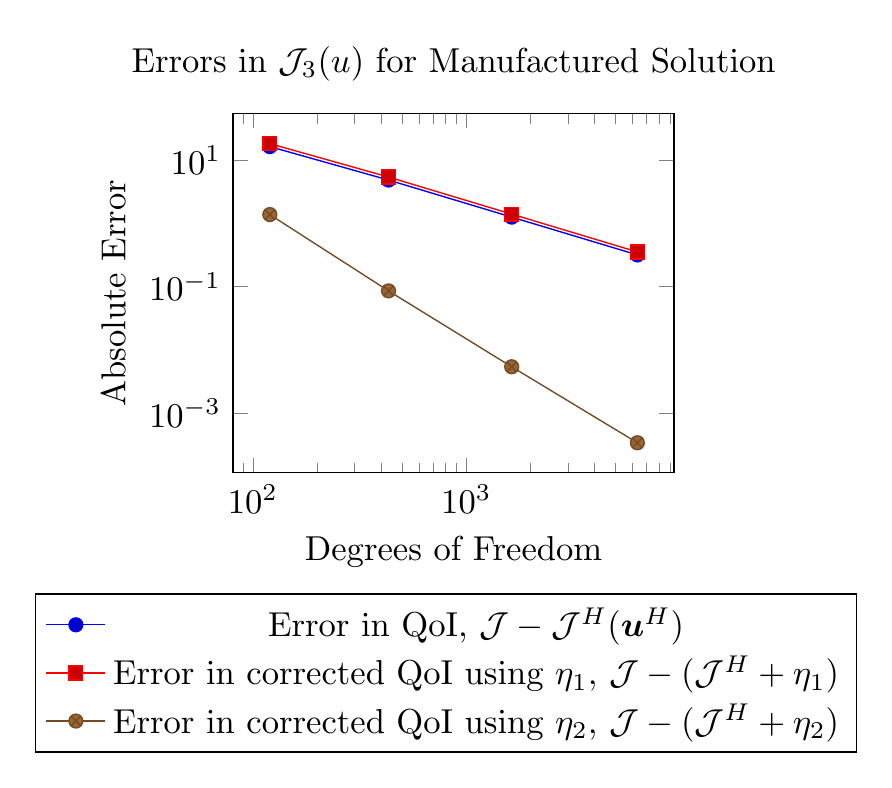}
\end{subfigure}
\caption{Asymptotic behavior of the estimates $\eta_1$ and $\eta_2$
for the QoI $\qoi_3(u)$ for the manufactured solution when $\alpha=10^{-2}$
on a sequence of uniformly refined meshes.}
\label{fig:qoi3_manufactured_plots}
\end{figure}

\begin{figure}[ht!]
\centering
\begin{subfigure}{0.52\textwidth}
\centering
\includegraphics[width=.99\linewidth]{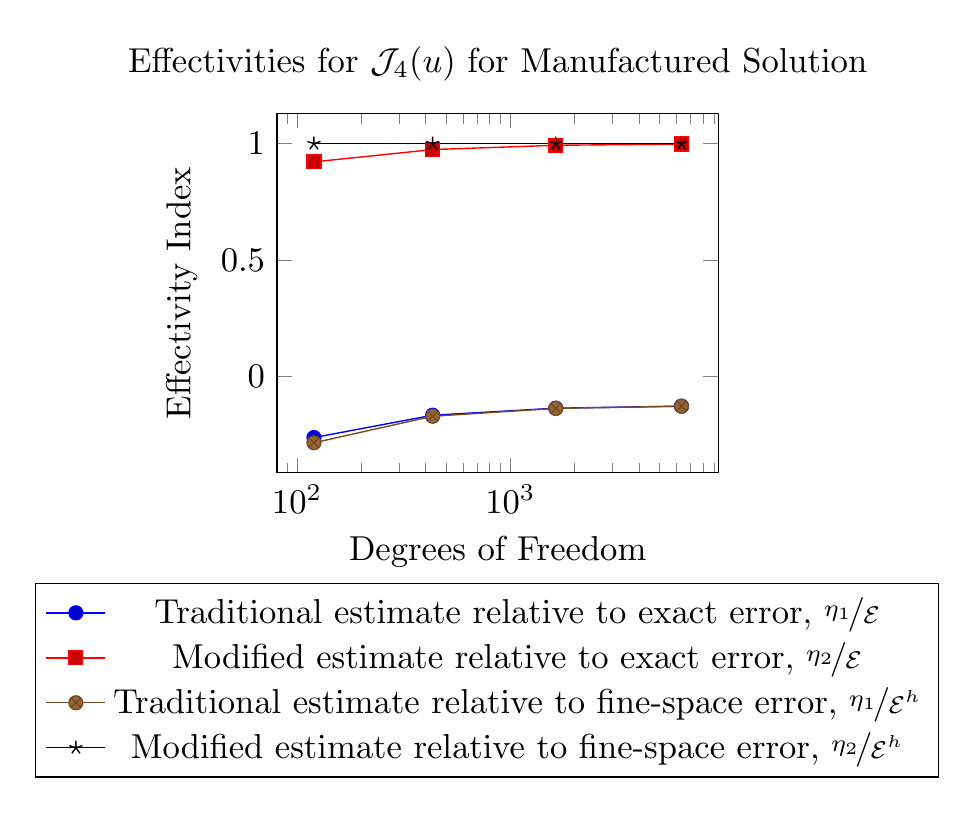}
\end{subfigure}%
\begin{subfigure}{0.47\textwidth}
\centering
\includegraphics[width=.99\linewidth]{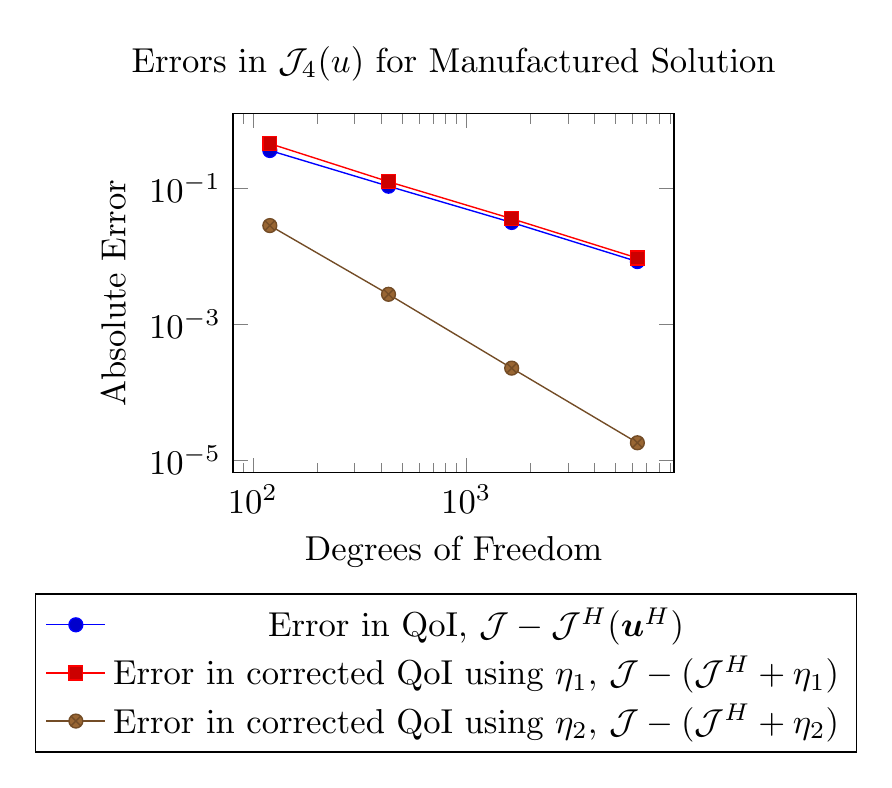}
\end{subfigure}
\caption{Asymptotic behavior of the estimates $\eta_1$ and $\eta_2$ for the
QoI $\qoi_4(u)$ for the manufactured solution when $\alpha=10^{-2}$ on a
sequence of uniformly refined meshes.}
\label{fig:qoi4_manufactured_plots}
\end{figure}

However, Figures~\ref{fig:qoi3_manufactured_plots} and
\ref{fig:qoi4_manufactured_plots} surprisingly illustrate that only the
newly proposed estimate $\eta_2$ is effective as the mesh size goes to
zero for the QoIs $\qoi_3(u)$ and $\qoi_4(u)$ for the chosen problem.
In fact, the error approximated by the traditional adjoint-weighted residual
estimate $\eta_1$ incurs a sign error. As a result, the corrected functional
using $\eta_1$ is actually \emph{less accurate} than that of the coarse-space
functional evaluation $\qoi^H(\bs{u}^H)$, as seen by the right-hand plots of
Figures~\ref{fig:qoi3_manufactured_plots} and
\ref{fig:qoi4_manufactured_plots}. In contrast, the newly proposed estimate
$\eta_2$ still results in a corrected functional evaluation that converges
at a faster rate than the coarse-space functional evaluation for the
QoIs $\qoi_3(u)$ and $\qoi_4(u)$.

\begin{table}[ht!]
\begin{center}
\begin{tabular}{c || c }
$n_{el}$ & $||\bs{E}^{\residop}_L||_2$ \\ \hline \hline
192 &   8.54316e+01 \\ \hline
768 &   1.76089e+01 \\ \hline
3072 &  3.05150e+00 \\ \hline
12287 & 4.26147e$-$01
\end{tabular}
\end{center}
\caption{Residual linearization errors for the manufactured solution
with $\alpha=10^{-2}$ at different mesh resolutions.}
\label{tab:manufactured_erl}
\end{table}

This leads to an insight of the present work: \emph{the effectiveness of
the traditional adjoint-weighted residual can be QoI dependent}. That is to
say, it appears that the QoI linearization error $\error^{\qoi}_L$ can have a
large impact on the effectiveness of the estimate $\eta_1$. We speculate
that this effect is typically unobserved in common goal-oriented error
estimation settings, where QoIs such as lift or drag over an airfoil are
\emph{linear} functionals, such that $\error^{\qoi}_L = 0$. Further, for our
specific PDE of interest,
the stark contrast in the asymptotic behavior of the estimate $\eta_1$ for the
QoI $\qoi_2(u)$ and the QoIs $\qoi_3(u)$ and $\qoi_4(u)$ indicates that the
residual linearization error $\bs{E}^{\residop}_L$ is not the culprit for this
loss of effectivity, since this term is independent of the chosen QoI.
That is to say, as the mesh size decreases, $\bs{E}^{\residop}_L \to \bs{0}$,
which is confirmed in Table~\ref{tab:manufactured_erl}.

\subsubsection{Mesh Adaptivity}
\label{sssec:manufactured_adapt}

In this section, we investigate the behavior of the nonlinear functional
quantities and the use of the localized error estimates $\eta^i_1$ and
$\eta^i_2$, $i=1,2,\dots,n_{vtx}$, in the adaptive procedure described in
Section~\ref{sec:adapt}. For all examples, we begin with an initial mesh
with $192$ elements, as shown in Figure~\ref{fig:results_domain_and_mesh}
and perform five iterations of the process:
\begin{equation*}
\text{solve primal PDE} \rightarrow
\text{solve adjoint PDE} \rightarrow
\text{estimate error} \rightarrow
\text{adapt mesh}
\end{equation*}
where either the traditional adjoint-weighted residual error estimate $\eta_1$
or the newly proposed error estimate $\eta_2$ is used to drive mesh
modification. During each adaptive iteration, the mesh size field is specified
according to equation~\eqref{eq:size_field} so that the desired number of
elements $T$ is twice the number of elements in the previous mesh, and
we evaluate the QoI $\qoi^H(\bs{u}^H)$ on the coarse space
and measure its error $\error$ with respect to the exact QoI value $\qoi(u)$.
Additionally, we compute the effectivity of both error estimates,
$\nicefrac{\eta_1}{\error}$ and $\nicefrac{\eta_2}{\error}$, respectively.
Figures~\ref{fig:manufactured_qoi2_adapt_convergence_plots},
\ref{fig:manufactured_qoi3_adapt_convergence_plots}, and
\ref{fig:manufactured_qoi4_adapt_convergence_plots} illustrate the
behavior of the adaptive schemes when $\eta_1$ and $\eta_2$ are chosen
to drive mesh adaptivity for the QoIs $\qoi_2(u), \qoi_3(u),$ and
$\qoi_4(u)$, respectively. The plot on the left of these figures shows
the effectivities of the estimates $\eta_1$ and $\eta_2$ as the mesh
is adapted. The plot on the right demonstrates the convergence behavior
of the error $\error$ as the mesh is adapted. Additionally, the plot on
the right of these figures compares the convergence behavior of the adaptive
schemes to that of uniform refinement.

\begin{figure}[ht!]
\centering
\begin{subfigure}{0.47\textwidth}
\centering
\includegraphics[width=.99\linewidth]{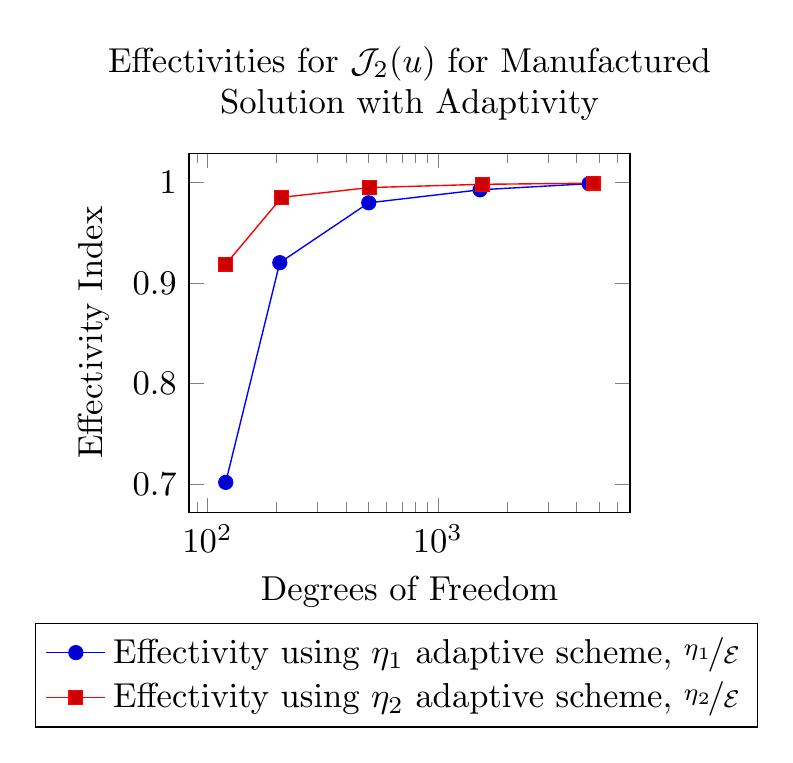}
\end{subfigure}%
\begin{subfigure}{0.43\textwidth}
\centering
\includegraphics[width=.99\linewidth]{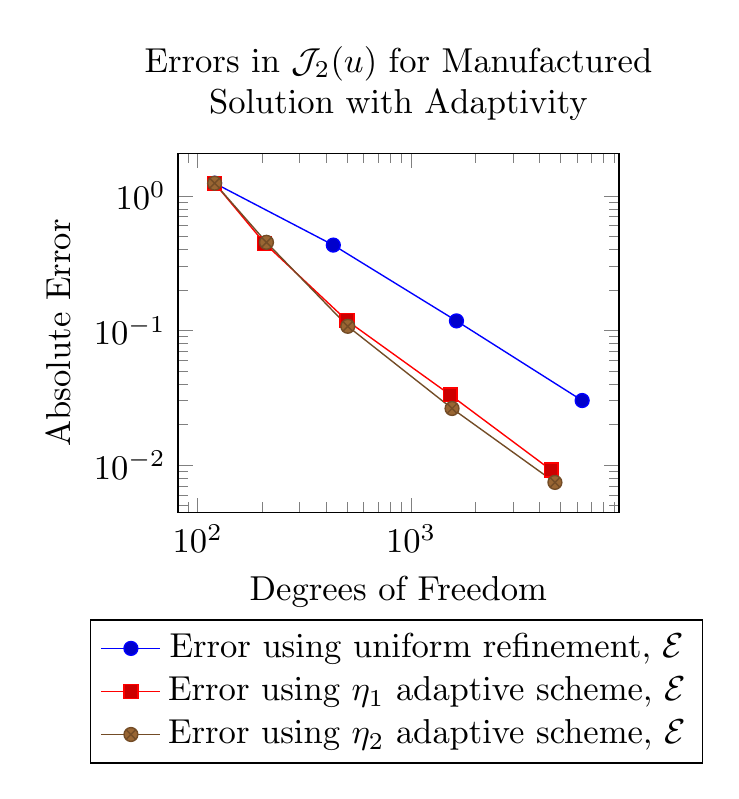}
\end{subfigure}
\caption{Behavior of the adaptive scheme using the estimates $\eta_1$ and
$\eta_2$ for the QoI $\qoi_2(u)$ for the manufactured solution when
$\alpha=10^{-2}$.}
\label{fig:manufactured_qoi2_adapt_convergence_plots}
\end{figure}

Figure~\ref{fig:manufactured_qoi2_adapt_convergence_plots} illustrates
that both estimates $\eta_1$ and $\eta_2$ are effective as the mesh is
adapted according to their localized indicators~\eqref{eq:eta1_localized}
and \eqref{eq:eta2_localized}, respectively, where the newly proposed
estimate $\eta_2$ is slightly more accurate at coarser mesh resolutions.
Despite this fact, however, the right-hand side of
Figure~\ref{fig:manufactured_qoi2_adapt_convergence_plots} illustrates that the
localizations of each estimate lead to adaptive schemes that are
very similar in terms of error convergence, while adapting based on the newly
proposed error estimate $\eta_2$ only marginally improves the accuracy of
the QoI when contrasted to adapting based on the traditional adjoint-weighted
residual error estimate $\eta_1$. Both adaptive schemes, though, considerably
outperform uniform refinement in terms of reducing the error at a fixed
number of degrees of freedom in the problem.

\begin{figure}[ht!]
\centering
\begin{subfigure}{0.47\textwidth}
\centering
\includegraphics[width=.99\linewidth]{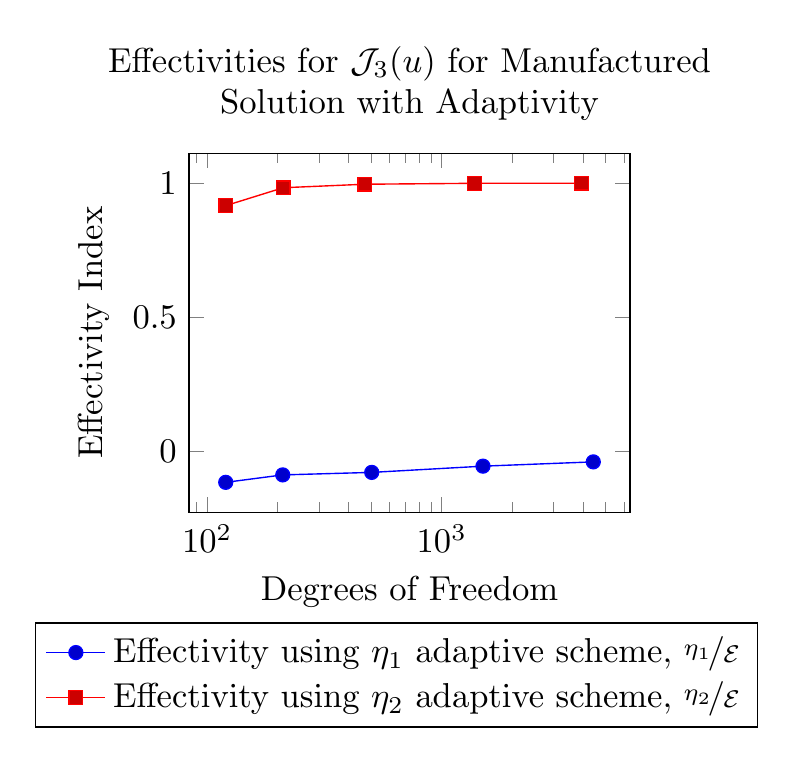}
\end{subfigure}%
\begin{subfigure}{0.43\textwidth}
\centering
\includegraphics[width=.99\linewidth]{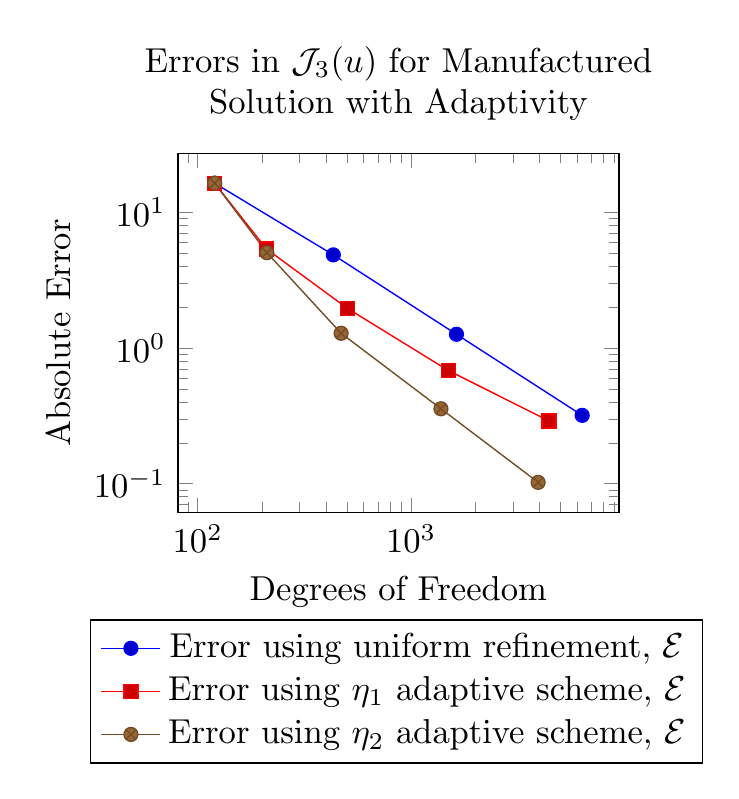}
\end{subfigure}
\caption{Behavior of the adaptive scheme using the estimates $\eta_1$ and
$\eta_2$ for the QoI $\qoi_3(u)$ for the manufactured solution when
$\alpha=10^{-2}$.}
\label{fig:manufactured_qoi3_adapt_convergence_plots}
\end{figure}

\begin{figure}[ht!]
\centering
\begin{subfigure}{0.47\textwidth}
\centering
\includegraphics[width=.99\linewidth]{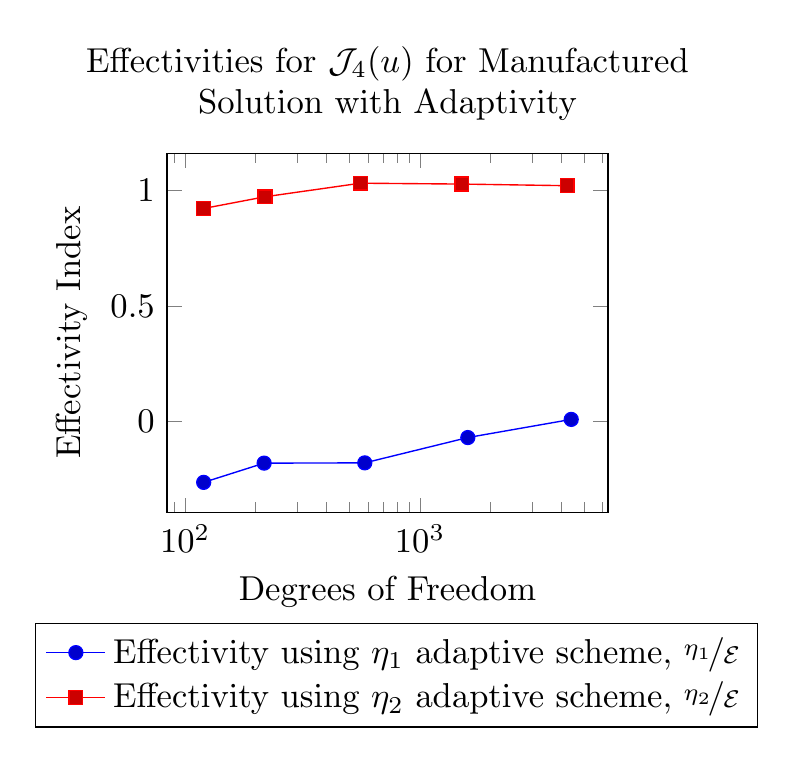}
\end{subfigure}%
\begin{subfigure}{0.43\textwidth}
\centering
\includegraphics[width=.99\linewidth]{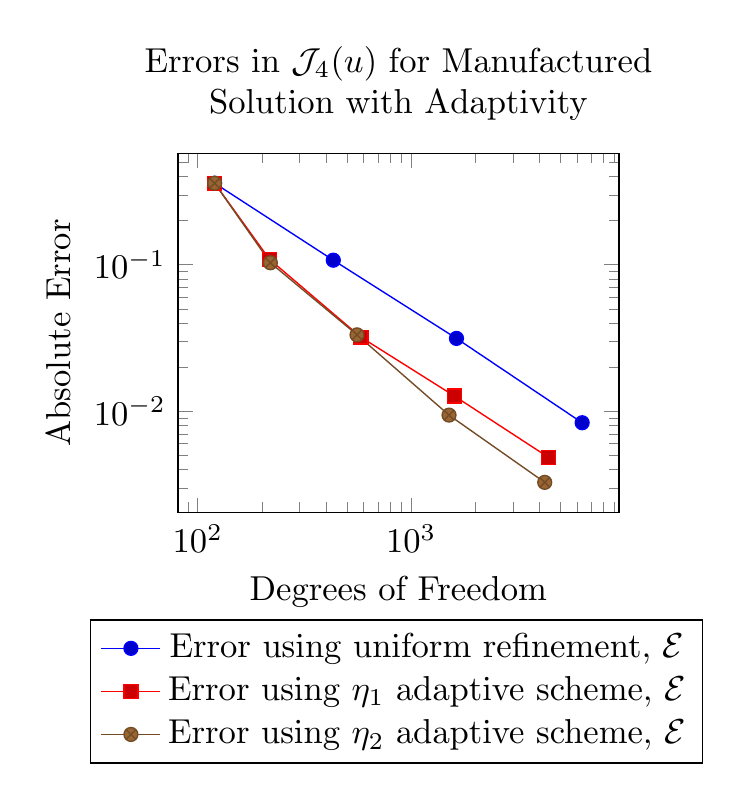}
\end{subfigure}
\caption{Behavior of the adaptive scheme using the estimates $\eta_1$ and
$\eta_2$ for the QoI $\qoi_4(u)$ for the manufactured solution when
$\alpha=10^{-2}$.}
\label{fig:manufactured_qoi4_adapt_convergence_plots}
\end{figure}

In contrast, the left-hand plots of
Figures~\ref{fig:manufactured_qoi3_adapt_convergence_plots} and
\ref{fig:manufactured_qoi4_adapt_convergence_plots} highlight that only
the newly proposed estimate $\eta_2$ is effective for the QoIs $\qoi_3(U)$
and $\qoi_4(u)$ as adaptive iterations are performed. Again, the traditional
adjoint-weighted residual estimate $\eta_1$ greatly under-predicts the 
QoI discretization error for these two quantities of interest, which results
in an effectivity $\nicefrac{\eta_1}{\error}$ close to $0$. The right-hand
plot of Figure~\ref{fig:manufactured_qoi3_adapt_convergence_plots} also
illustrates that driving mesh adaptivity with the estimate $\eta_2$
leads to a significantly more accurate adaptive scheme for the QoI
$\qoi_3(u)$ when compared to driving adaptivity with the estimate $\eta_1$ or
when compared to uniform refinement. The right-hand plot of
Figure~\ref{fig:manufactured_qoi4_adapt_convergence_plots} tells a similar story
for the QoI $\qoi_4(u)$, however, the improvement of the adaptive scheme using
the estimate $\eta_2$ over the estimate $\eta_1$ is not as pronounced for
this QoI.

This leads to another insight of the present work: \emph{incorporating
localized information about the linearization errors $\bs{E}^{\residop}_L$ and
$\error^{\qoi}_L$ into an adaptive procedure can lead to more optimal
meshes in certain scenarios}. For our particular chosen PDE, the fact that
$\eta_1$ and $\eta_2$ lead to adaptive schemes that provide quite similar
behavior for the QoI $\qoi_2(u)$ but not for $\qoi_3(u)$ and $\qoi_4(u)$
suggests that localizing the QoI linearization $\error^{\qoi}_L$ may provide
more benefit than localizing the residual linearization error
$\bs{E}^{\residop}_L$. We have presently made no effort to decouple these
effects, however, as an avenue for future work one could investigate the
approximation $\error^h \approx -\bs{z}^* \cdot \bs{R}^h(\bs{u}^H_h)$, which
does not account for the residual linearization error $\bs{E}^{\residop}_L$.
In general, the interaction between the two linearization error terms will
depend on the structure of the primal PDE and the chosen QoI.

\subsubsection{A Problem with Gradient Singularities}
\label{ssec:singular}

In this section, we consider the model problem~\eqref{eq:strong_form}
over the domain $\Omega$ shown in Figure~\ref{fig:results_domain_and_mesh}
with a forcing function $f$ chosen to be $f=100$ and the scaling parameter
$\alpha=10^{-2}$. For context, Newton's method takes $4$ iterations for the
solution of the primal problem on the coarse space and the fine space for
all results shown below. The choice of forcing function $f$ and domain
$\Omega$ leads to gradient singularities in the solution $u$ at each
interior corner of the domain $\Omega$. In a general sense, without
accurately resolving these singularities, so-called ``pollution'' 
error~\cite{babuvska1994pollution} will affect the accuracy of the finite element
approximation to the QoI throughout the domain. Figure~\ref{fig:singular_solution}
illustrates the solution $u$ and its gradient
magnitude $|\!| \nabla u |\!|$ of the model problem with the chosen
constant forcing function.

\begin{figure}[ht!]
\centering
\begin{subfigure}{0.33\textwidth}
\centering
\includegraphics[width=.99\linewidth]{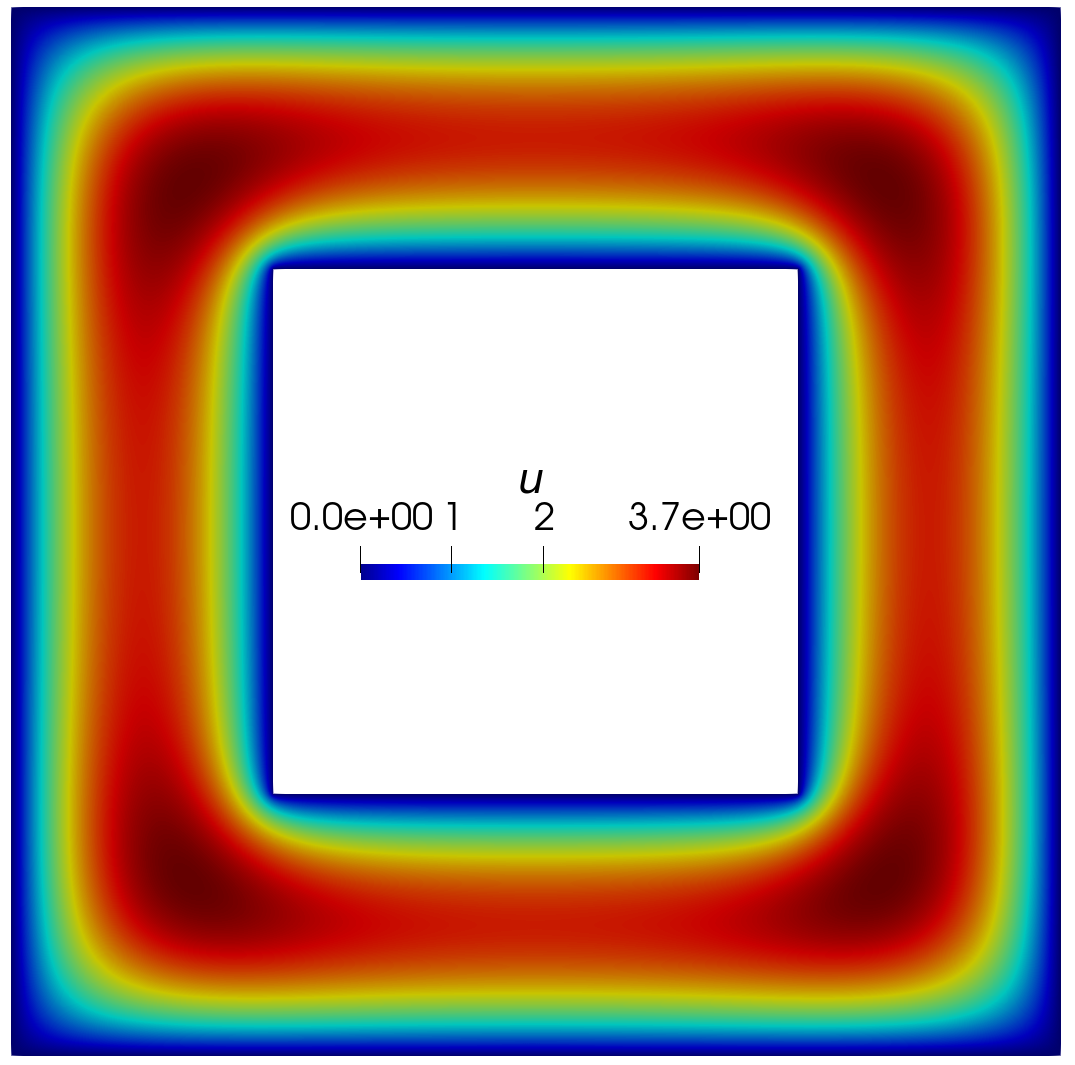}
\end{subfigure}%
\begin{subfigure}{0.33\textwidth}
\centering
\includegraphics[width=.99\linewidth]{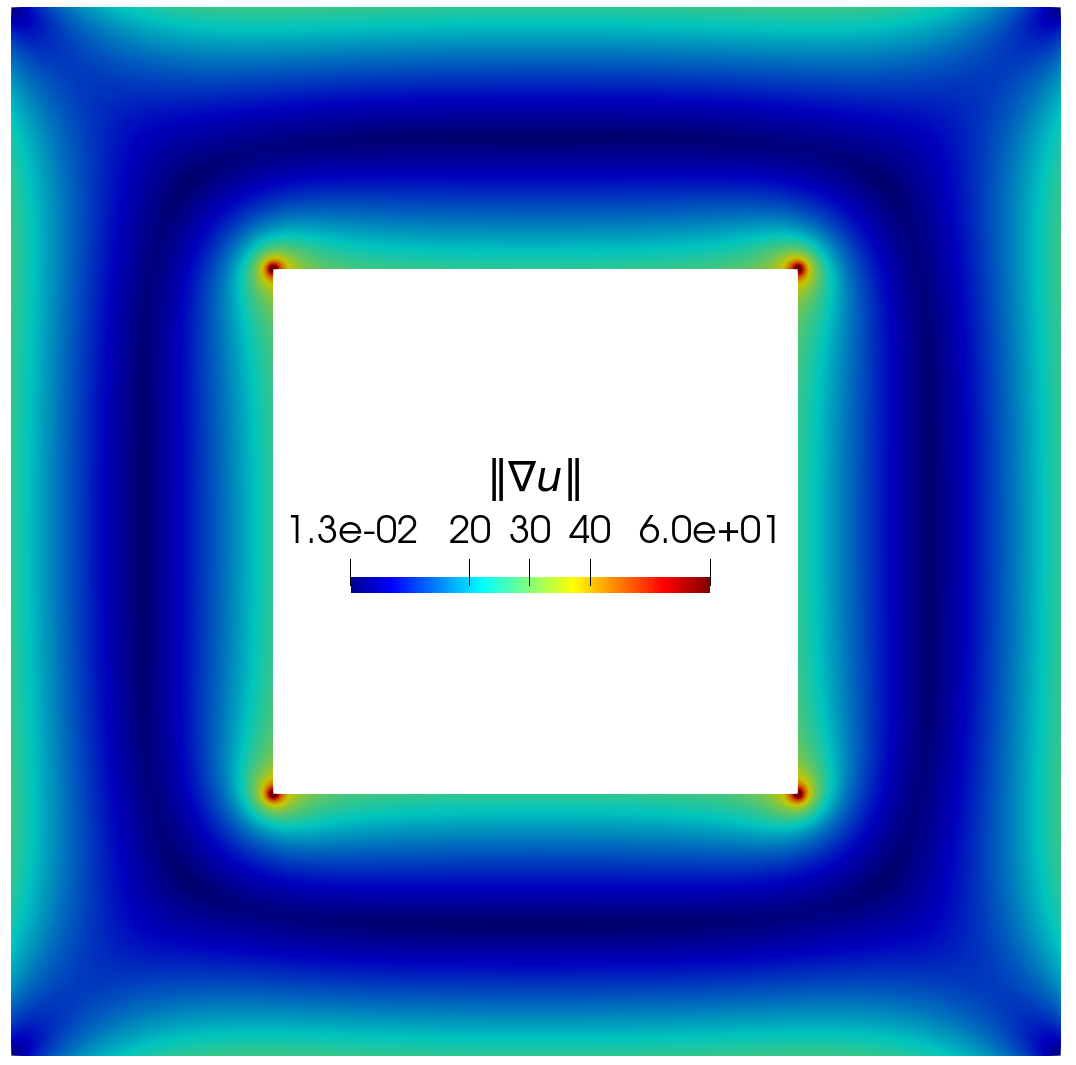}
\end{subfigure}
\caption{The solution with gradient singularities (left), the $L_2$
norm of the solution gradient (right).}
\label{fig:singular_solution}
\end{figure}

We investigate the behavior of all chosen functional quantities
$\qoi_i(u), i=1,2,3,4,$ and the use of the localized error estimates
$\eta^i_1$ and $\eta^i_2, i = 1, 2, \dots, n_{vtx}$, in the adaptive
process described in Section~\ref{sec:adapt}. In general, the exact values
of the QoIs are not quantifiable analytically. To approximate these exact
values, the primal problem is solved on a mesh with about one half
million degrees of freedom using the fine space $\fspace^h$. This mesh
is finer at every spatial location than the final adapted meshes in the
following results, and the reference values computed by this process
for each QoI are shown in~\ref{sec:qoi_values_singular}.

\begin{figure}[ht!]
\centering
\begin{subfigure}{0.45\textwidth}
\centering
\includegraphics[width=.99\linewidth]{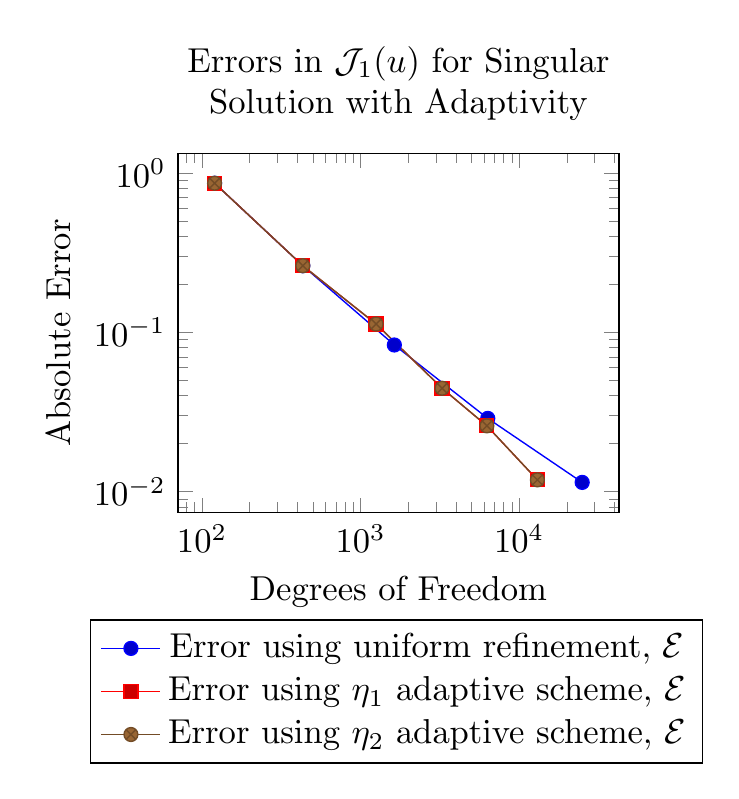}
\end{subfigure}%
\begin{subfigure}{0.45\textwidth}
\centering
\includegraphics[width=.99\linewidth]{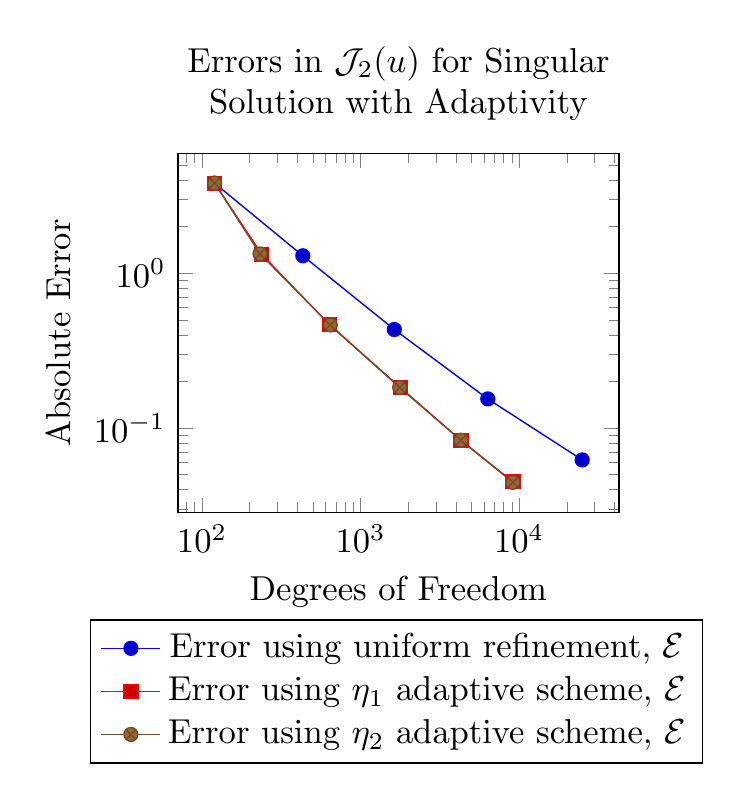}
\end{subfigure}
\caption{Convergence behavior for the example problem with gradient
singularities using uniform refinement and adaptivity driven by the
estimates $\eta_1$ and $\eta_2$ for the $\qoi_1(u)$ (left) and
$\qoi_2(u)$ (right).}
\label{fig:singular_adapt_convergence_plots1}
\end{figure}

For all examples, we begin with an initial mesh with $192$ elements, as shown
in Figure~\ref{fig:results_domain_and_mesh} and perform five iterations
of the process:
\begin{equation*}
\text{solve primal PDE} \rightarrow
\text{solve adjoint PDE} \rightarrow
\text{estimate error} \rightarrow
\text{adapt mesh}
\end{equation*}
where either the traditional adjoint-weighted residual error estimate
$\eta_1$ or the newly proposed estimate $\eta_2$ is used to drive mesh
modification. At each adaptive iteration, the mesh size field is specified
according to equation~\eqref{eq:size_field} so that the target number of
elements $T$ in the resultant mesh is twice the number of elements in the
current mesh. At each mesh instance, we evaluate the chosen QoI
$\qoi^H(\bs{u}^H)$ on the coarse space and measure its error $\error$
with respect to the reference value that approximates its exact value
$\qoi(u)$.

\begin{figure}[ht!]
\centering
\begin{subfigure}{0.45\textwidth}
\centering
\includegraphics[width=.99\linewidth]{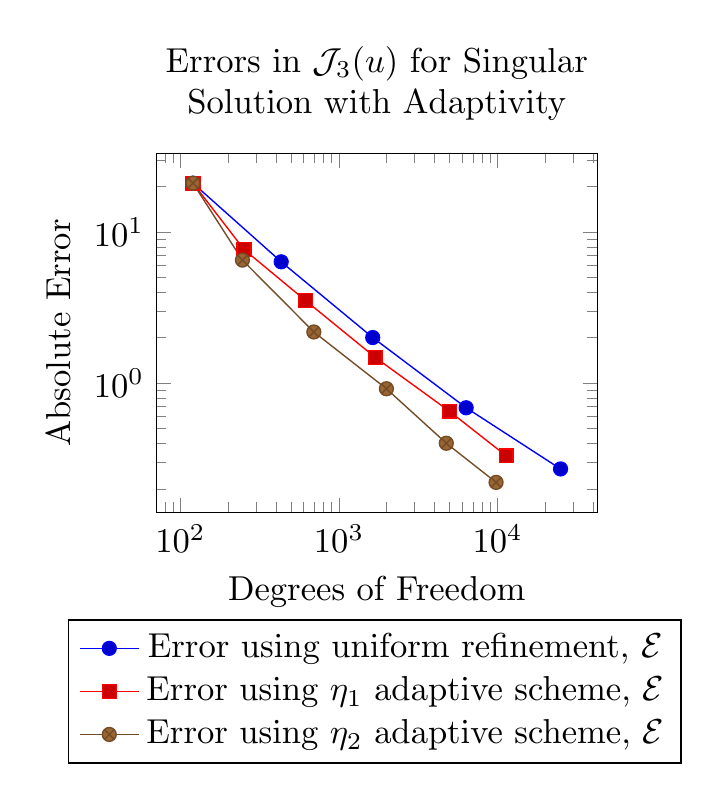}
\end{subfigure}%
\begin{subfigure}{0.45\textwidth}
\centering
\includegraphics[width=.99\linewidth]{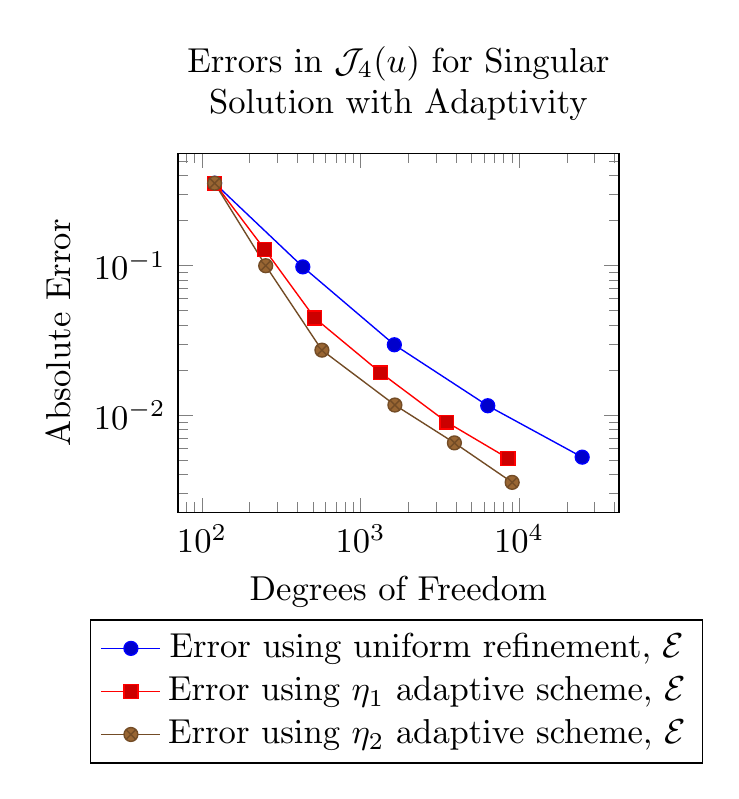}
\end{subfigure}
\caption{Convergence behavior for the example problem with gradient
singularities using uniform refinement and adaptivity driven by the
estimates $\eta_1$ and $\eta_2$ for the $\qoi_3(u)$ (left) and
the $\qoi_4(u)$ (right).}
\label{fig:singular_adapt_convergence_plots2}
\end{figure}

Figure~\ref{fig:singular_adapt_convergence_plots1}
illustrates the convergence behavior of the QoI discretization error
when the mesh is uniformly refined and when $\eta_1$ and $\eta_2$ are
chosen to drive mesh adaptivity for $\qoi_1(u)$ and $\qoi_2(u)$.
Figure~\ref{fig:singular_adapt_convergence_plots2} similarly illustrates
this behavior for $\qoi_3(u)$ and $\qoi_4(u)$. Interestingly,
Figure~\ref{fig:singular_adapt_convergence_plots1} shows that both estimates,
$\eta_1$ and $\eta_2$, result in adaptive convergence histories with
identical errors. This is because both schemes lead to identical meshes
at each adaptive iteration, as illustrated in
Figure~\ref{fig:adapt_singular_mesh_qoi2}, where the resultant mesh of the
third, fourth, and fifth adaptive iteration is shown. In addition, the
plot on the left of Figure~\ref{fig:singular_adapt_convergence_plots1}
illustrates that both estimates, $\eta_1$ and $\eta_2$, only begin to provide
improved convergence behavior at finer mesh resolutions, whereas the plot on
the right of Figure~\ref{fig:singular_adapt_convergence_plots1} demonstrates
that both estimates improve the QoI convergence behavior, where the error is
reduced by over half an order of magnitude at the final mesh resolution of
around $10,000$ degrees of freedom. This may, in part, be explained by the
fact that $\qoi_1(u)$ is a \emph{global} QoI, while $\qoi_2(u)$ is
restricted to only a localized portion of the domain.

\begin{figure}[ht]
\centering
\begin{subfigure}{0.33\textwidth}
\centering
\includegraphics[width=.99\linewidth]{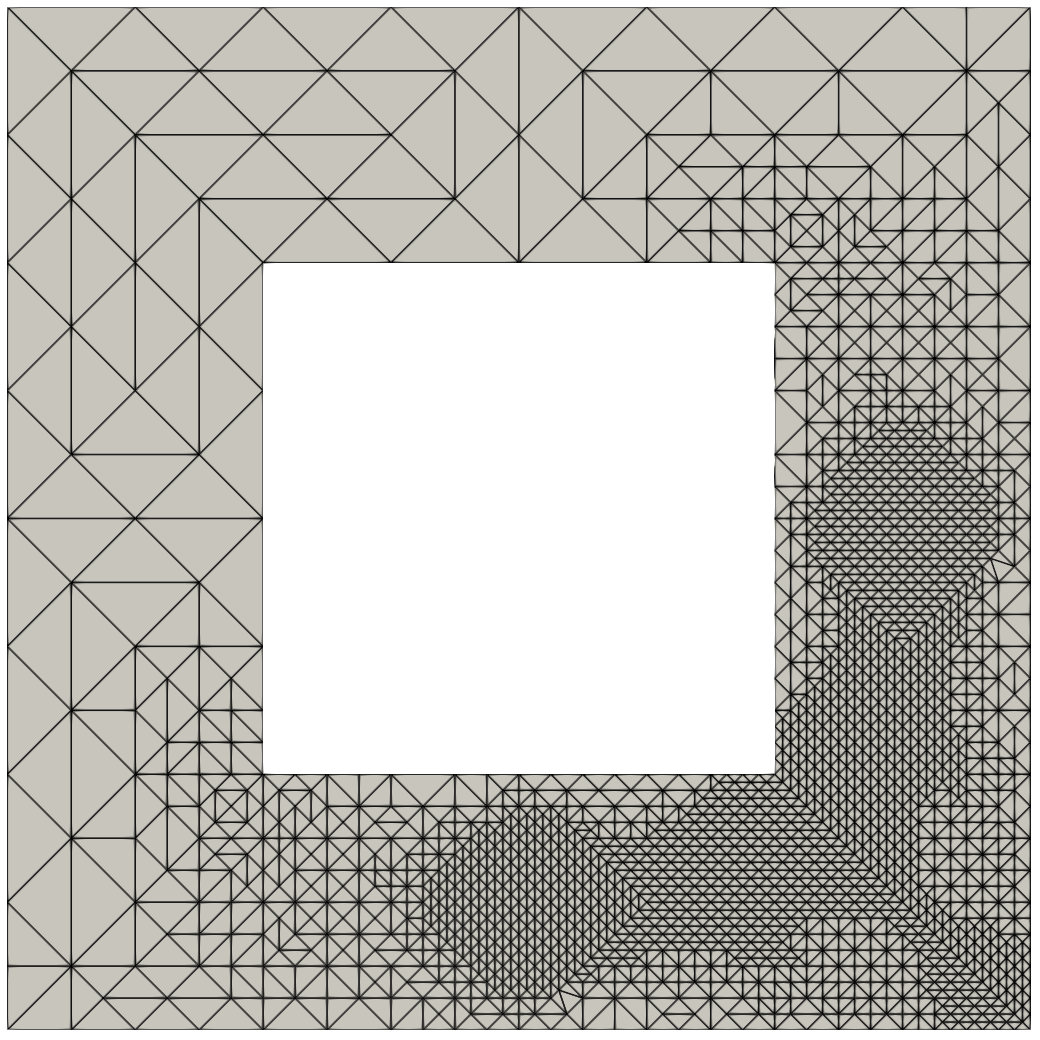}
\end{subfigure}%
\begin{subfigure}{0.33\textwidth}
\centering
\includegraphics[width=.99\linewidth]{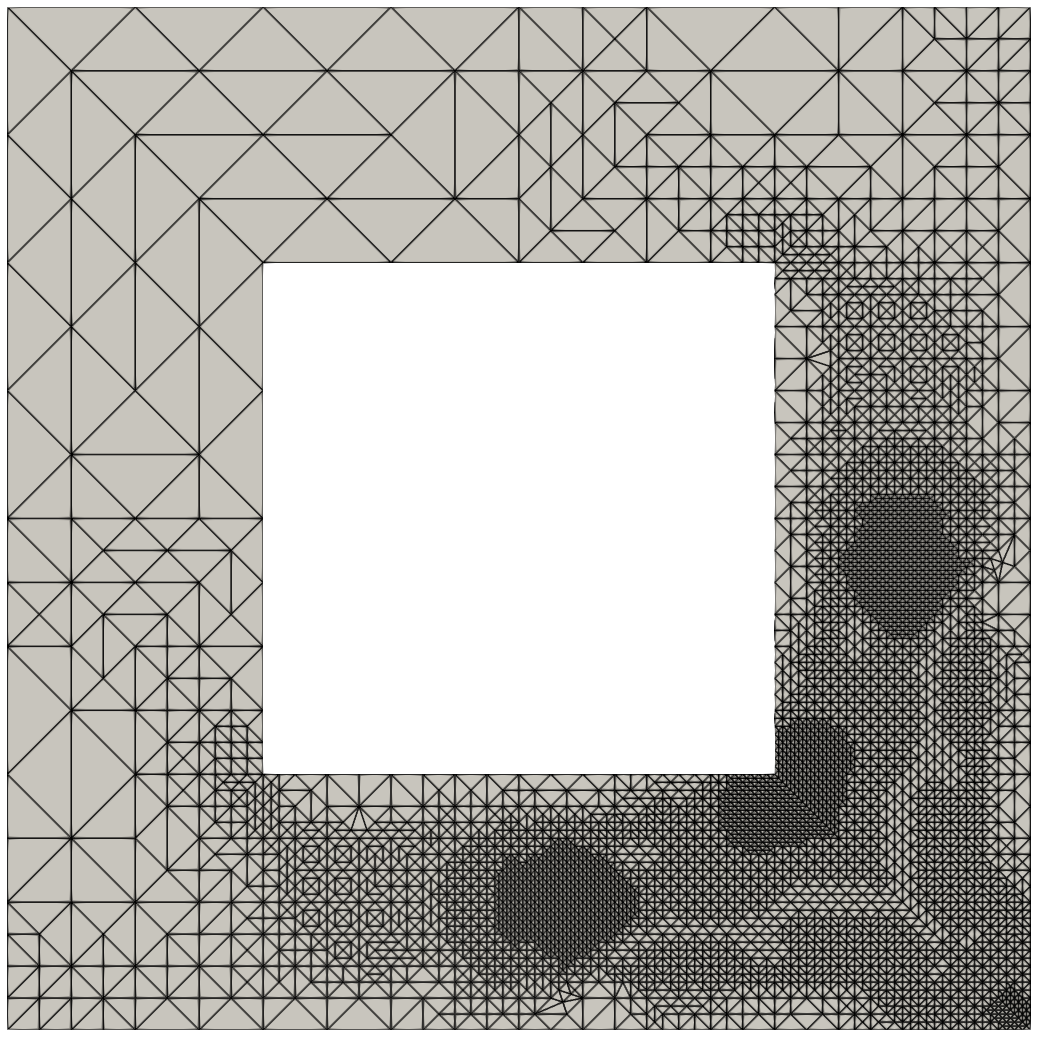}
\end{subfigure}%
\begin{subfigure}{0.33\textwidth}
\centering
\includegraphics[width=.99\linewidth]{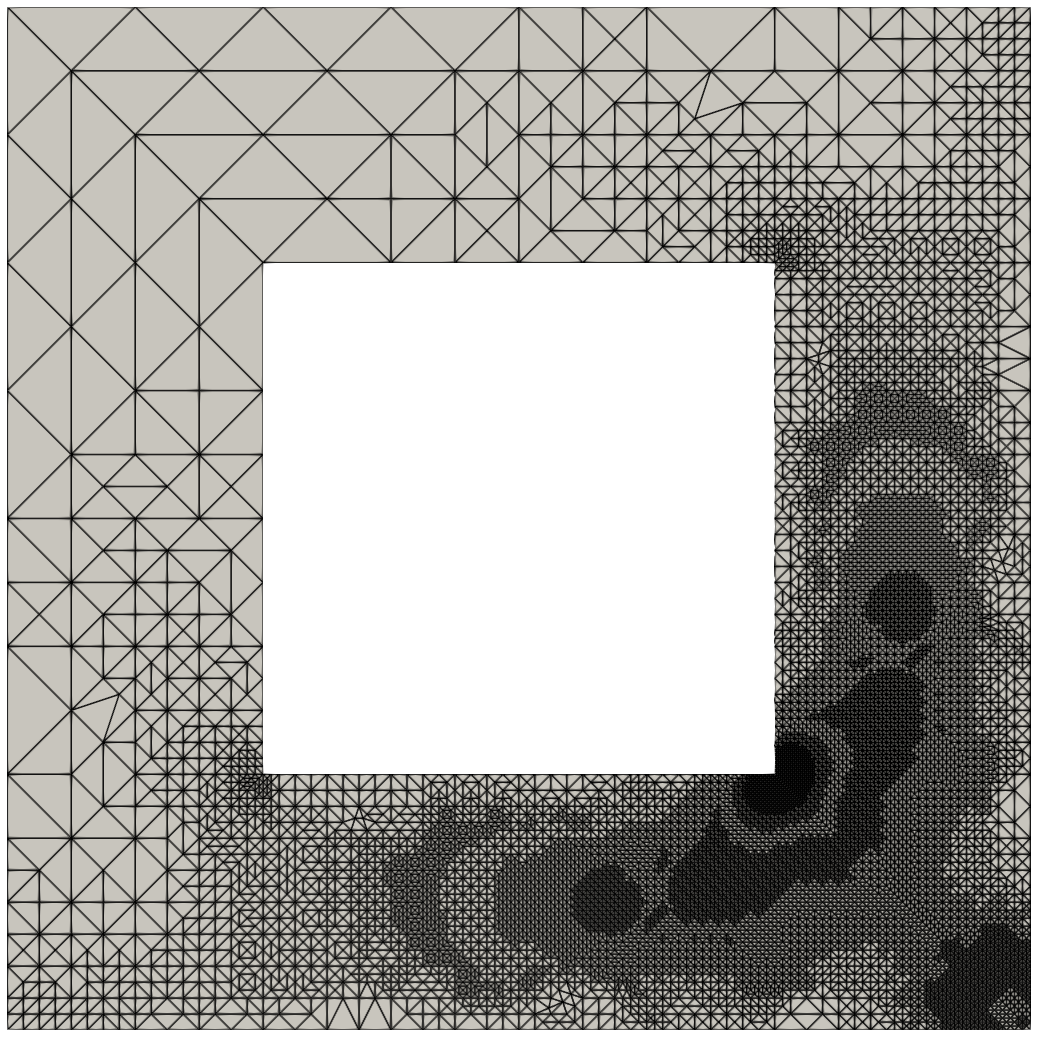}
\end{subfigure}
\caption{Sequence of meshes obtained for the example problem with gradient
singularities for $\qoi_2(u)$ using either the estimate $\eta_1$
or $\eta_2$ at the third (left), fourth (center), and fifth (right)
adaptive iterations.}
\label{fig:adapt_singular_mesh_qoi2}
\end{figure}

Figure~\ref{fig:singular_adapt_convergence_plots2} highlights that, for the
QoIs $\qoi_3(u)$ and $\qoi_4(u)$, the error convergence behavior using
either the estimate $\eta_1$ or $\eta_2$ is improved when compared to
uniform refinement, but that using the estimate $\eta_2$ provides a
further reduction in error when compared to $\eta_1$. This result is
consistent with the behavior seen in the previous manufactured solution
section. Additionally, Figures~\ref{fig:adapt_singular_mesh_qoi3} and
\ref{fig:adapt_singular_mesh_qoi4} compare the final meshes obtained
from the iterative adaptive process when considering $\eta_1$ or $\eta_2$
for the QoIs $\qoi_3(u)$ and $\qoi_4(u)$, respectively. Here, we make
several remarks. First, for both QoIs, both estimates heavily refine the
gradient singularity nearest to the local QoI subdomain, refine to a
lesser degree the next two nearest gradient singularities, and
do not resolve at all the furthest gradient singularity. Second,
comparing the image on the right to the image on the left of these figures
highlights that both estimates lead to qualitatively distinct 
meshes. This can, in part, be explained by the fact that the newly
proposed estimate $\eta_2$ includes linearization errors in its localization,
while the traditional adjoint-weighted residual $\eta_1$ does not.

\begin{figure}[ht]
\centering
\begin{subfigure}{0.33\textwidth}
\centering
\includegraphics[width=.99\linewidth]{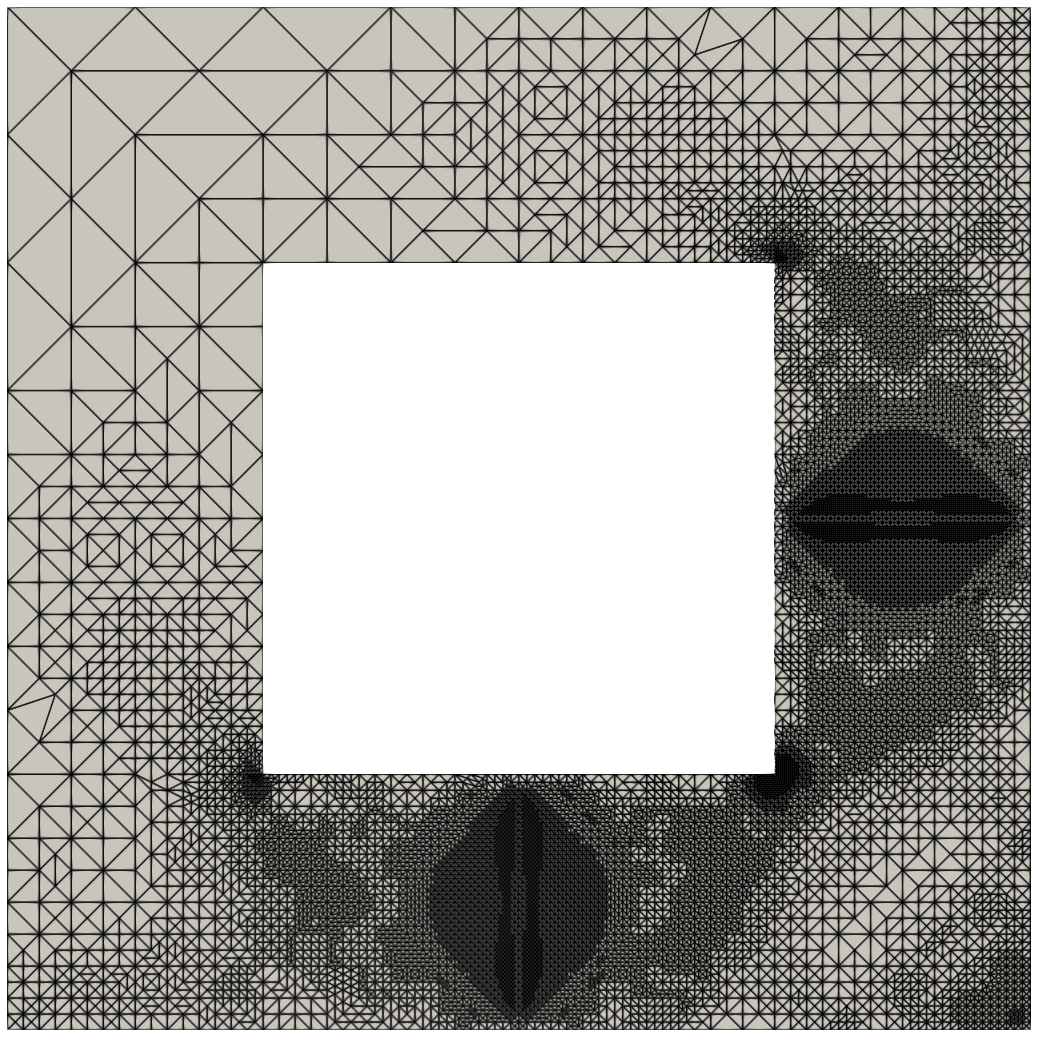}
\end{subfigure}%
\begin{subfigure}{0.33\textwidth}
\centering
\includegraphics[width=.99\linewidth]{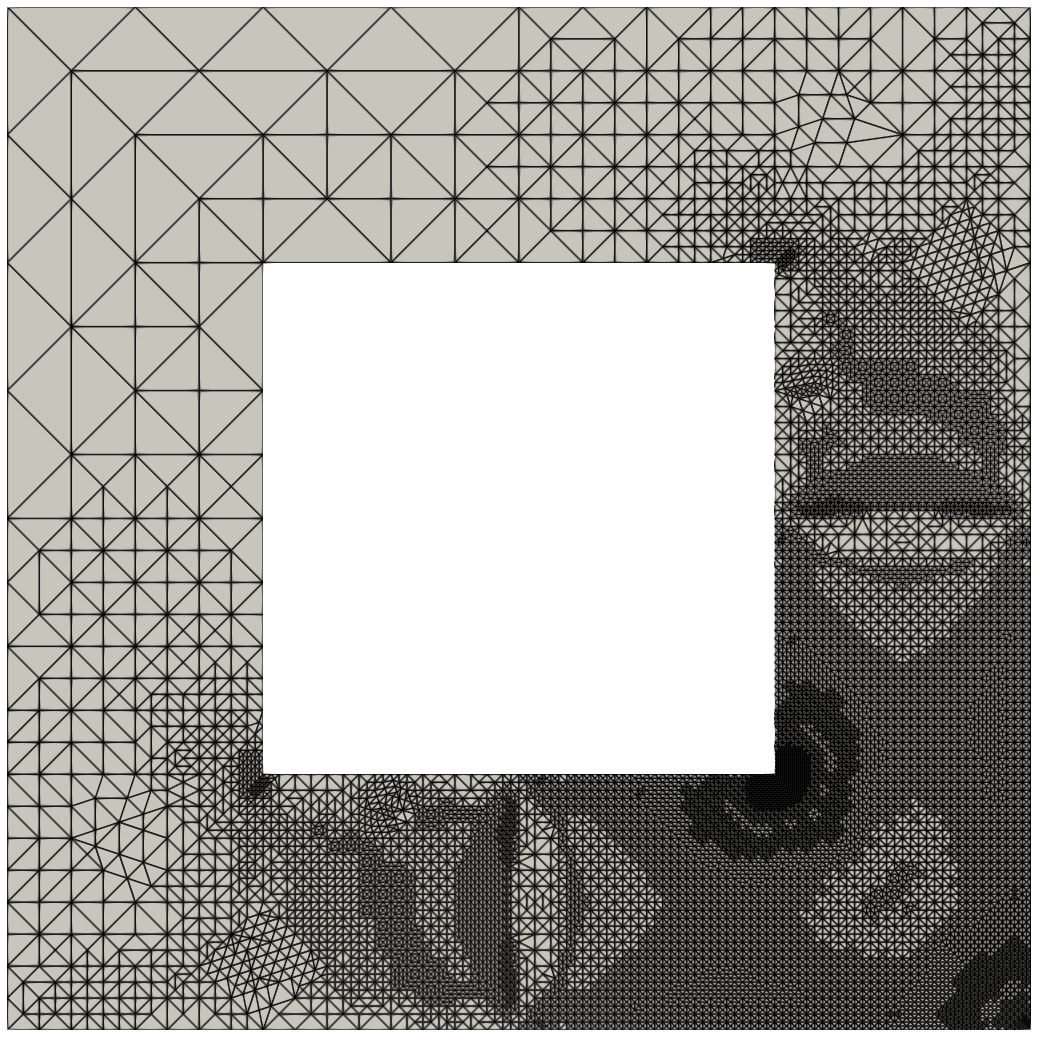}
\end{subfigure}
\caption{Final adapted meshes obtained for the example problem with gradient
singularities for $\qoi_3(u)$ using the estimate $\eta_1$ (left) and
the estimate $\eta_2$ (right).}
\label{fig:adapt_singular_mesh_qoi3}
\end{figure}

\begin{figure}[ht]
\centering
\begin{subfigure}{0.33\textwidth}
\centering
\includegraphics[width=.99\linewidth]{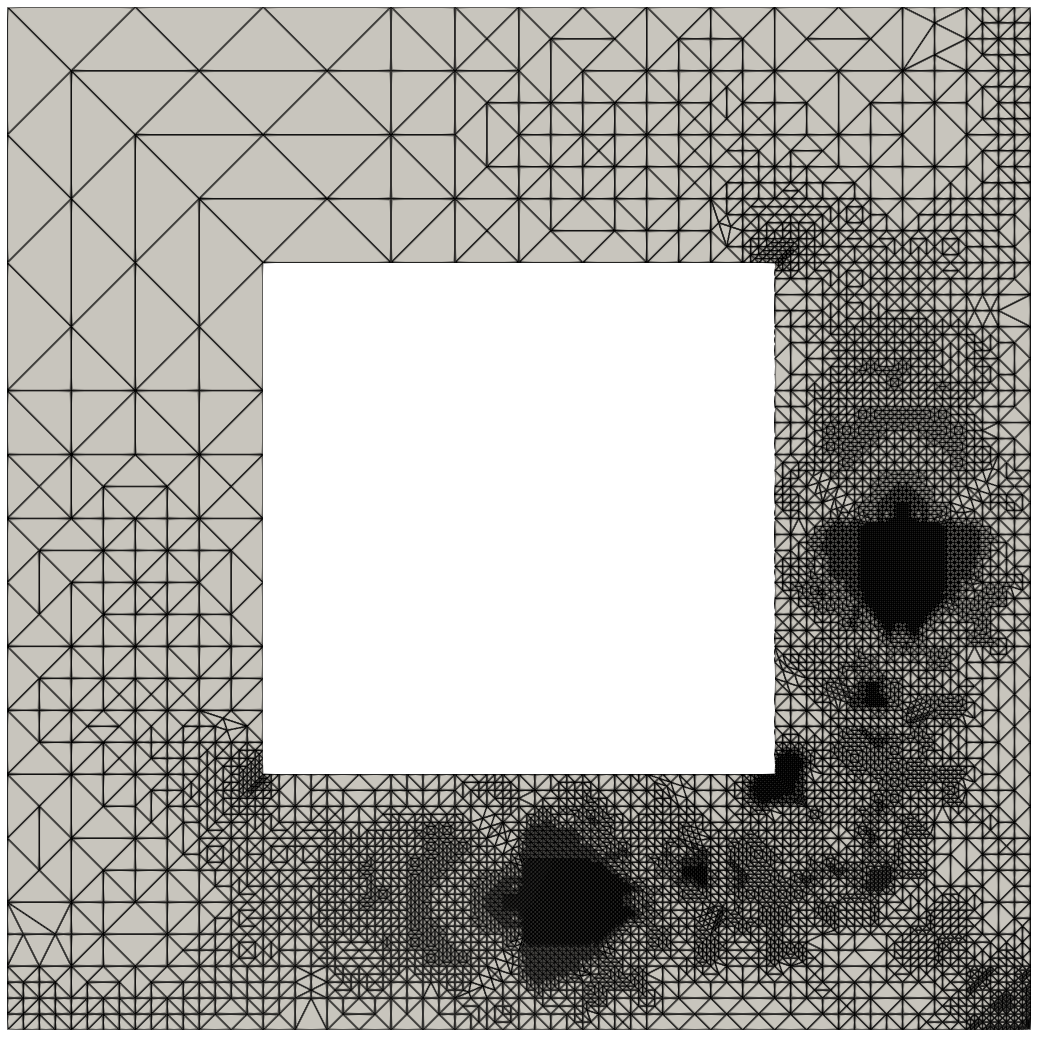}
\end{subfigure}%
\begin{subfigure}{0.33\textwidth}
\centering
\includegraphics[width=.99\linewidth]{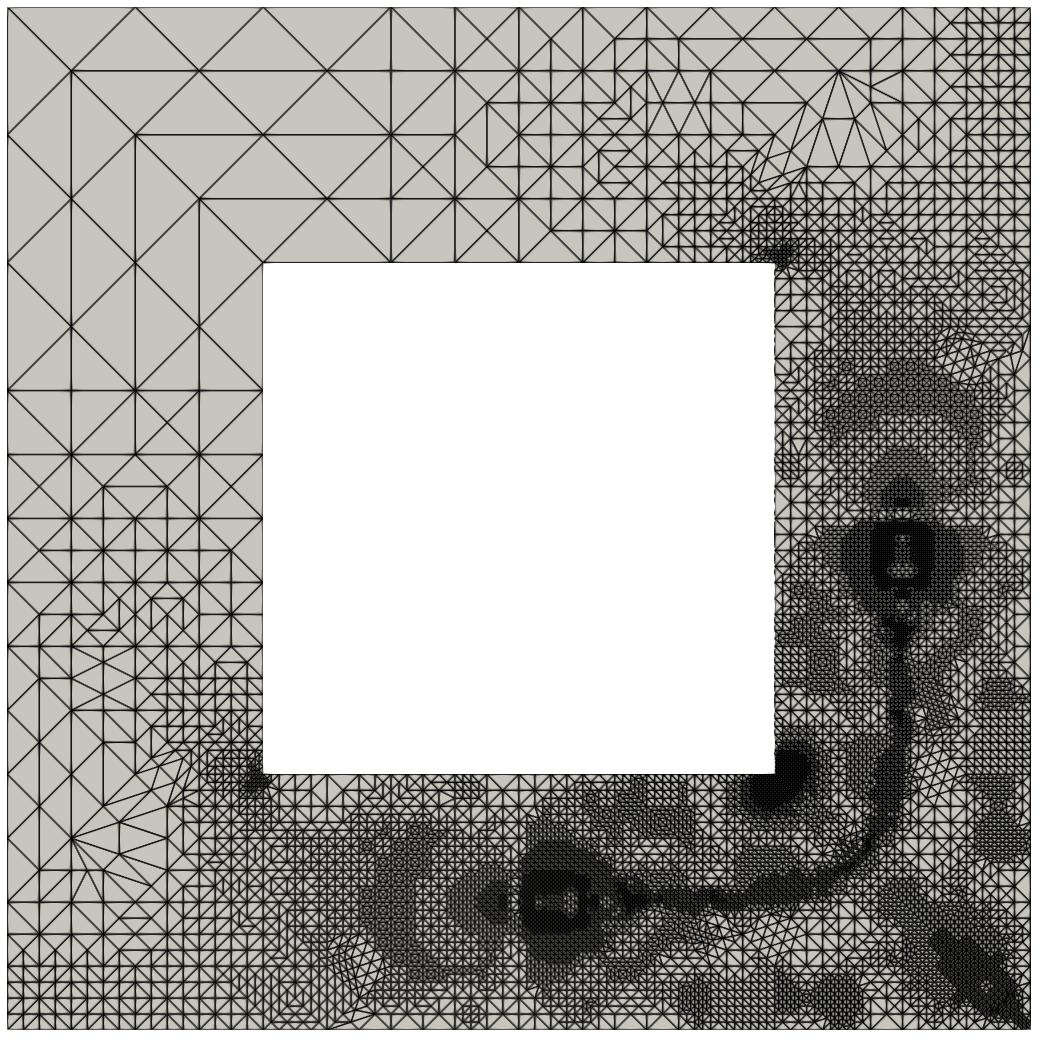}
\end{subfigure}
\caption{Final adapted meshes obtained for the example problem with gradient
singularities for the QoI $\qoi_4(u)$ using the estimate $\eta_1$ (left) and
the estimate $\eta_2$ (right).}
\label{fig:adapt_singular_mesh_qoi4}
\end{figure}

\subsection{Finite Deformation Elasticity}
\label{ssec:elasticity}

In this section, we consider finite deformation elasticity with a neo-Hookean
material model in three spatial dimensions. Let the domain boundary $\Gamma$
be decomposed into $\Gamma_G \subseteq \Gamma$ and $\Gamma_H = \Gamma \setminus \Gamma_G$. Let $\bs{X} \in \Omega \subset \reals^3$ denote a
 point in a reference configuration that transforms to a point
$\bs{x} \in \Omega_t$ in a deformed configuration $\Omega_t \subset \reals^3$
after the reference domain $\Omega$ undergoes some deformation.
Let $\bs{u} = \bs{x} - \bs{X}$ denote the displacement.
Let $\bs{F} := \bs{I} + \frac{\partial \bs{u}}{\partial \bs{X}}$ denote
the deformation gradient, where $\bs{I}$ denotes the second-order identity
tensor. Let $J := \det(\bs{F})$ denote the determinant of the deformation
gradient. In this context, the balance of linear momentum without inertial and
body forces can be written as
\begin{equation}
\begin{cases}
\begin{aligned}
-\nabla \cdot \bs{P} &= \bs{0}, \quad && \bs{X} \in \Omega, \\
\bs{u} &= \bs{G}, && \bs{X} \in \Gamma_G, \\
\bs{P} \cdot \bs{N} &= \bs{0}, && \bs{X} \in \Gamma_H.
\end{aligned}
\end{cases}
\label{eq:elast_strong_form}
\end{equation}
Here, $\bs{G}$ denotes an externally applied displacement, $\bs{N}$ denotes
the unit outward normal to the boundary $\Gamma_H$,
$\bs{P} := J \bs{\sigma} \bs{F}^{-T}$ denotes the first Piola-Kirchhoff stress
tensor, and $\bs{\sigma}$ denotes the Cauchy stress tensor.
We consider a neo-Hookean constitutive model determined by the relationship
\begin{equation}
\bs{\sigma} = \mu J^{-\nicefrac53} \text{dev}(\bs{F}\bs{F}^{T}) +
\frac{\kappa}{2}(J - 1/J) \bs{I},
\end{equation}
where $\kappa$ denotes the bulk modulus, $\mu$ denotes the shear modulus,
and $\text{dev}(\cdot)$ denotes the deviatoric portion of a second-order tensor.
Placing the problem \eqref{eq:elast_strong_form} in weak form yields the
semilinear form
\begin{equation}
\lhsop(\bs{u}; \bs{w}) := \int_{\Omega} \bs{P}(\bs{u}) :
\nabla \bs{w} \, \text{d} \Omega,
\end{equation}
and the linear functional $\rhsop(\bs{w}) = 0$.

\begin{figure}[ht]
\centering
\begin{subfigure}{0.33\textwidth}
\centering
\includegraphics[width=.99\linewidth]{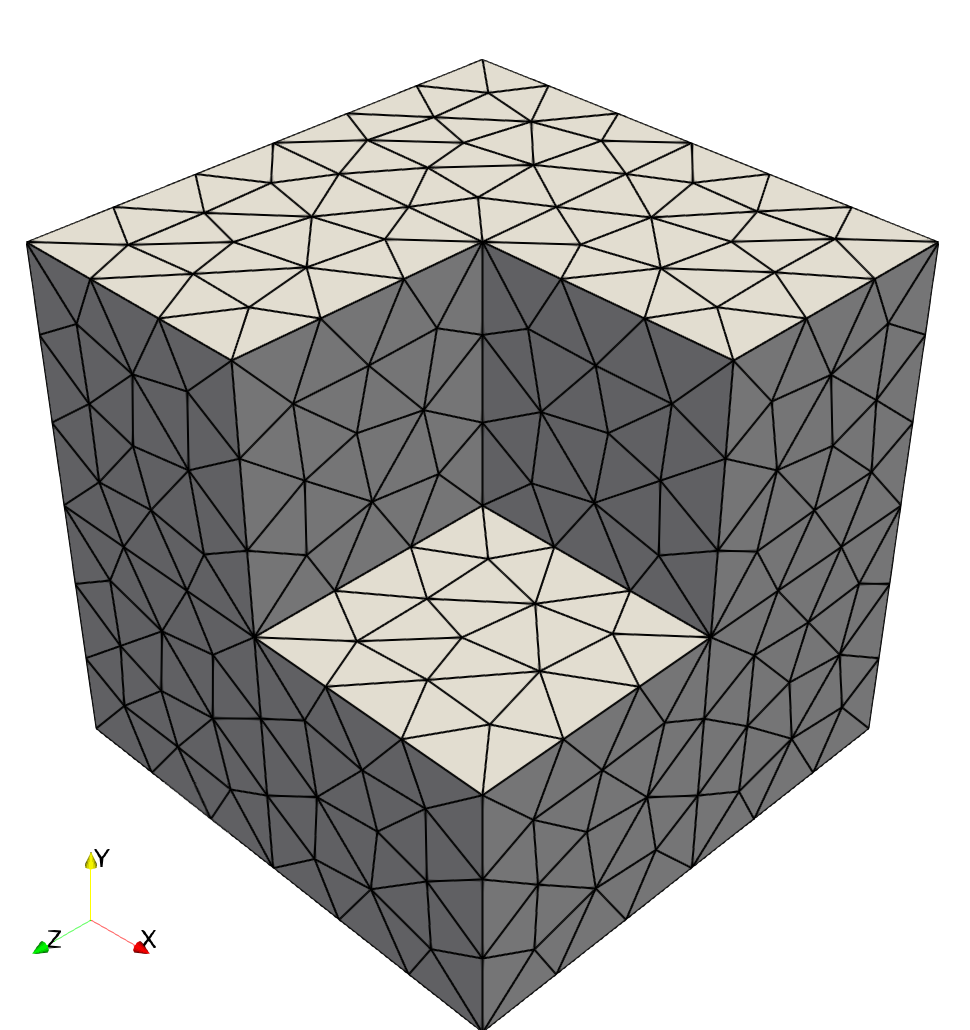}
\end{subfigure}%
\begin{subfigure}{0.33\textwidth}
\centering
\includegraphics[width=.99\linewidth]{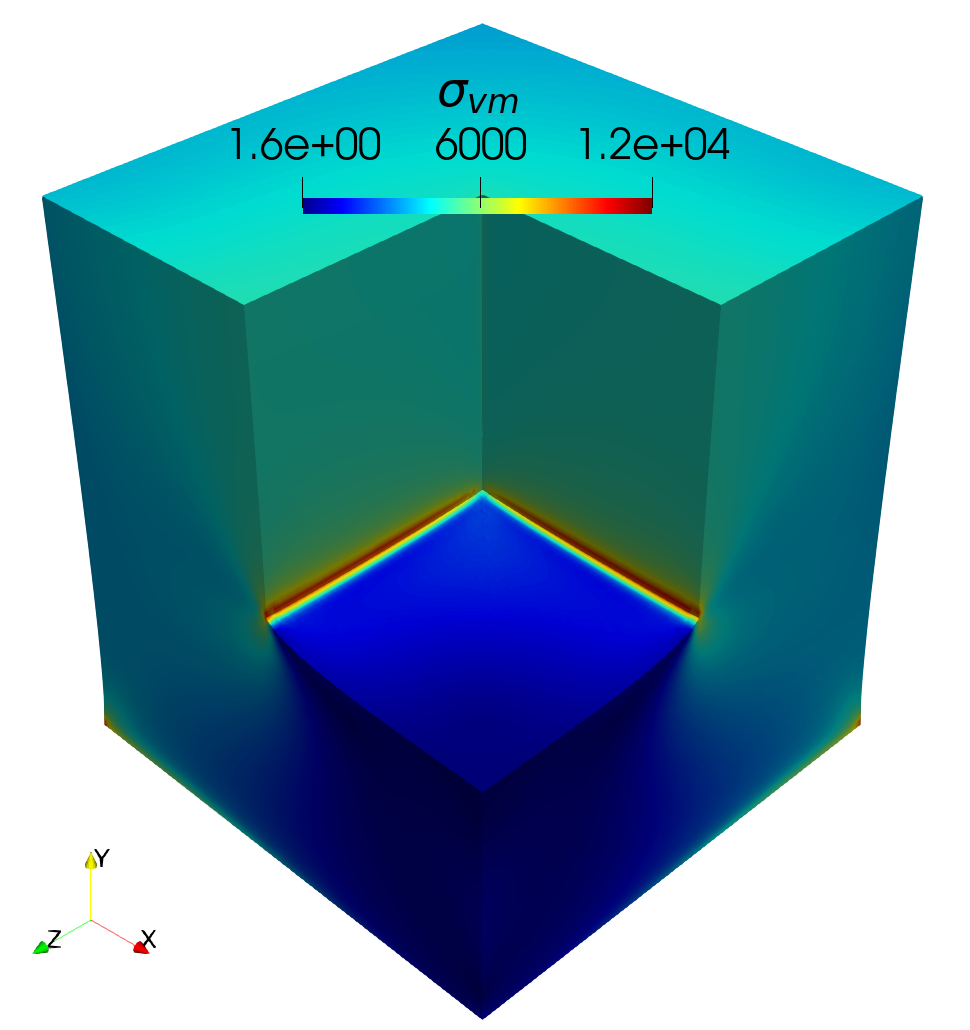}
\end{subfigure}%
\begin{subfigure}{0.33\textwidth}
\centering
\includegraphics[width=.99\linewidth]{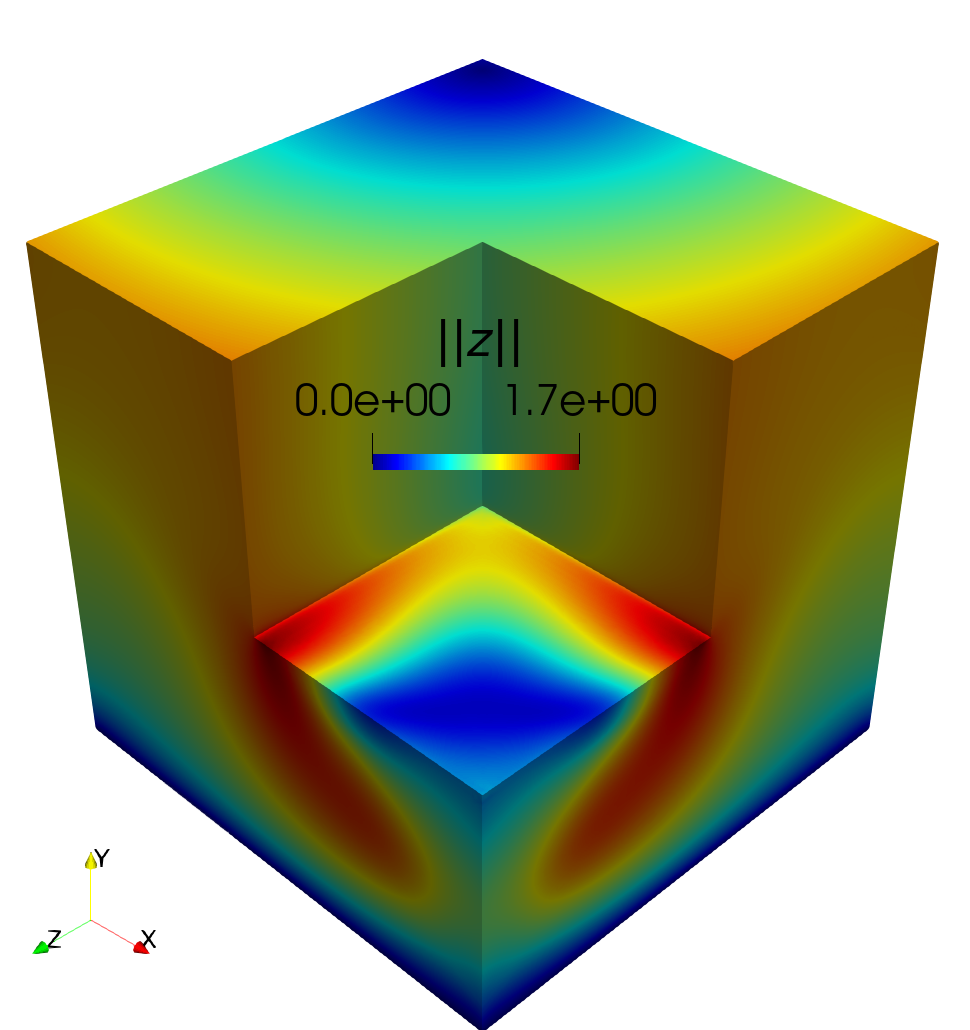}
\end{subfigure}
\caption{The domain $\Omega$ and its initial mesh for the nonlinear elasticity
example (left), the deformed domain $\Omega_t$ after loading (exaggerated by a
factor of 5) with the von Mises stress $\sigma_{vm}$ (MPa) plotted (center),
and the norm of the adjoint solution vector $\bs{z}$ (right).}
\label{fig:elasticity}
\end{figure}

As a model problem, we consider the domain $\Omega$ shown in Figure
\ref{fig:elasticity} which corresponds to a cube of dimensions
$5 \text{mm} \times 5 \text{mm} \times 5 \text{mm}$ with a
cube of size $2.5 \text{mm} \times 2.5 \text{mm} \times 2.5 \text{mm}$
missing from its uppermost corner. Displacements are set to
zero, $\bs{u} = \bs{0}$, on the minimal $y$-face of the geometry and the
problem is driven by a prescribed $y$-displacement of $u_y = 0.1 \text{mm}$,
which corresponds to a $2\%$ strain in the $y$-direction. We consider
a material with elastic modulus $E = 192.7$GPa and Poisson's
ratio $\nu=0.27$, from which the bulk and shear moduli can be determined
by the relationships $\kappa = \frac{E}{3(1-2 \nu)}$ and
$\mu = \frac{E}{2(1 + \nu)}$, respectively.

We again consider solving the finite element problem corresponding to the
governing equations on a coarse space with piecewise linear Lagrange basis
functions and on a fine space with piecewise quadratic Lagrange basis
functions. As a quantity of interest, we consider the von Mises stress
$\sigma_{vm}$ integrated over the domain
\begin{equation}
\qoi_{vm}(\bs{u}) := \int_{\Omega} \, \sigma_{vm} \; \text{d} \Omega,
\end{equation}
where $\sigma_{vm} := \sqrt{\frac32 s_{ij} s_{ij}}$ (with summation
on repeated indices implied), and $\bs{s}$ denotes
the deviatoric portion of the Cauchy stress tensor, as defined by
$\bs{s} := \mu J^{-\nicefrac53} \text{dev}(\bs{F}\bs{F}^{T})$.

\begin{figure}[ht!]
\centering
\begin{subfigure}{0.47\textwidth}
\centering
\includegraphics[width=.99\linewidth]{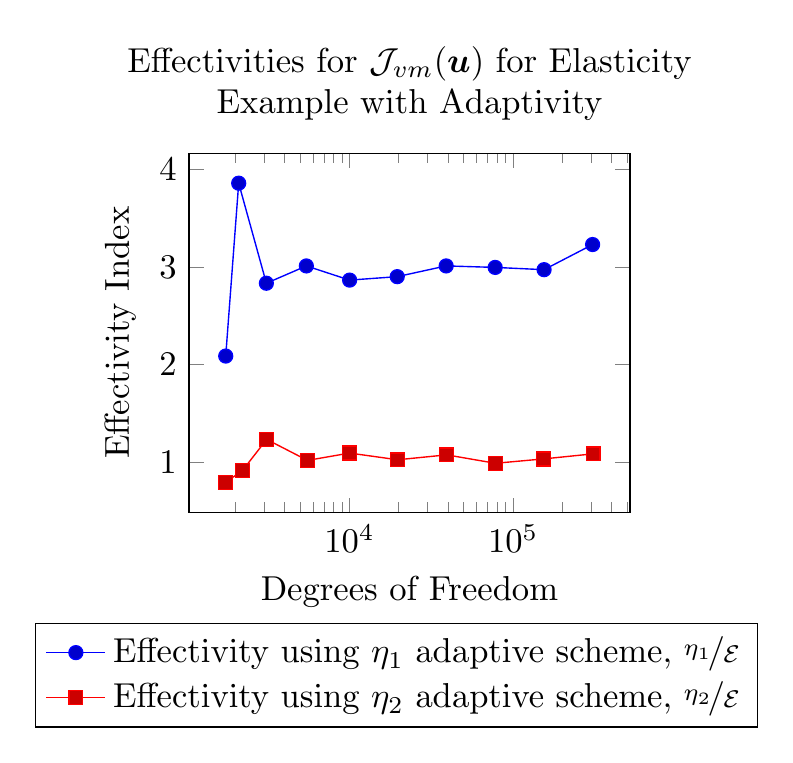}
\end{subfigure}%
\begin{subfigure}{0.43\textwidth}
\centering
\includegraphics[width=.99\linewidth]{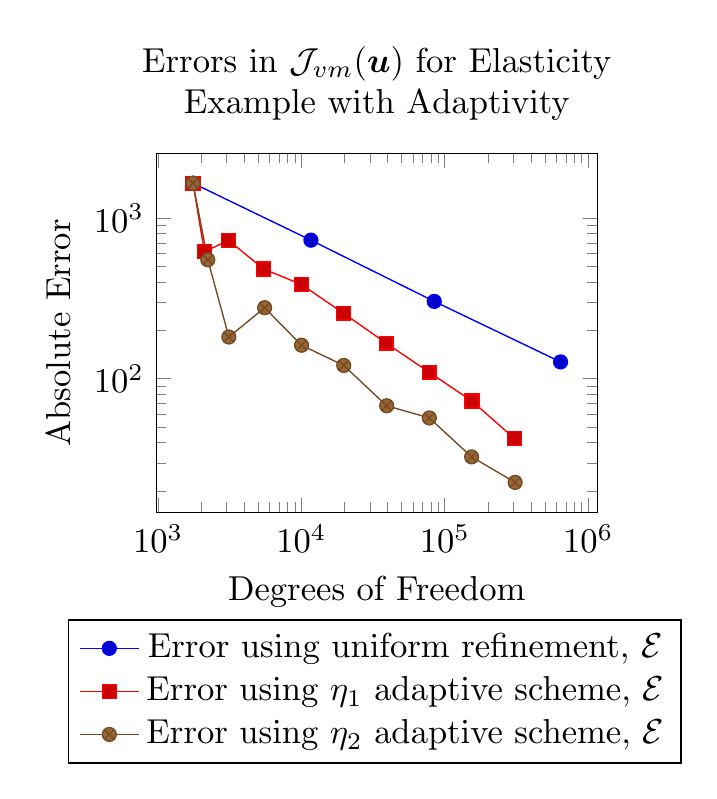}
\end{subfigure}
\caption{Behavior of the adaptive scheme using the estimates $\eta_1$
and $\eta_2$ for the QoI $\qoi_{vm}(u)$ for the finite deformation
elasticity example.}
\label{fig:elasticity_plots}
\end{figure}

We consider the initial mesh of the domain $\Omega$ shown in the
left-most image of Figure~\ref{fig:elasticity} and perform
$10$ iterations of the process:
\begin{equation*}
\text{solve primal PDE} \rightarrow
\text{solve adjoint PDE} \rightarrow
\text{estimate error} \rightarrow
\text{adapt mesh}
\end{equation*}
where either the traditional adjoint-weighted residual $\eta_1$
estimate or the newly proposed estimate $\eta_2$ is used to drive
mesh adaptivity. At each adaptive iteration, the mesh size field
is set according to equation~\eqref{eq:size_field} so that the target
number of elements $T$ in the resultant mesh is twice the number of
elements in the current mesh. At each mesh instance, we evaluate
the QoI $\qoi^H_e(\bs{u}^H)$ on the space and measure its error
with respect to a reference approximation for the exact value of
the QoI, compute an approximation to the exact error $\error$ using
this reference value, and compute the effectivity indices
$\eta_1/\error$ and $\eta_2/\error$.
To find the reference value, the primal problem is solved
on a mesh with $17.3$ million degrees of freedom using the fine
space $\fspace^h$, where this mesh is finer at every spatial location
than the final adapted meshes used in the results below, which contain
about $300,000$ degrees of freedom. The
reference value is found to be $\qoi_{vm}(\bs{u}) \approx 3.731860\text{GPa}$.

Figure~\ref{fig:elasticity_plots} illustrates the behavior of
the error estimates $\eta_1$ and $\eta_2$ as adaptive iterations
are performed. The left-hand side of Figure~\ref{fig:elasticity_plots}
illustrates illustrates that the newly proposed estimate
$\eta_2$ is more effective for the QoI $\qoi_{vm}(\bs{u})$ than
the traditional adjoint-weighted residual estimate $\eta_1$.
Note that for this problem the error in the QoI is negative, which is why
the effectivity index for $\eta_1$ is consistently greater than one.
This is in contrast to the examples in the previous section.
The right-hand side of Figure~\ref{fig:elasticity_plots} highlights
once again that the localizations of each estimate lead to adaptive
schemes that are different in terms of convergence behavior.
Specifically, driving mesh adaptivity with the estimate $\eta_2$
leads to a significantly more accurate adaptive scheme for the
QoI $\qoi_{vm}(\bs{u})$ when compared to driving adaptivity
with the estimate $\eta_1$ or when compared to uniform
refinement.

\begin{figure}[ht]
\centering
\begin{subfigure}{0.33\textwidth}
\centering
\includegraphics[width=.99\linewidth]{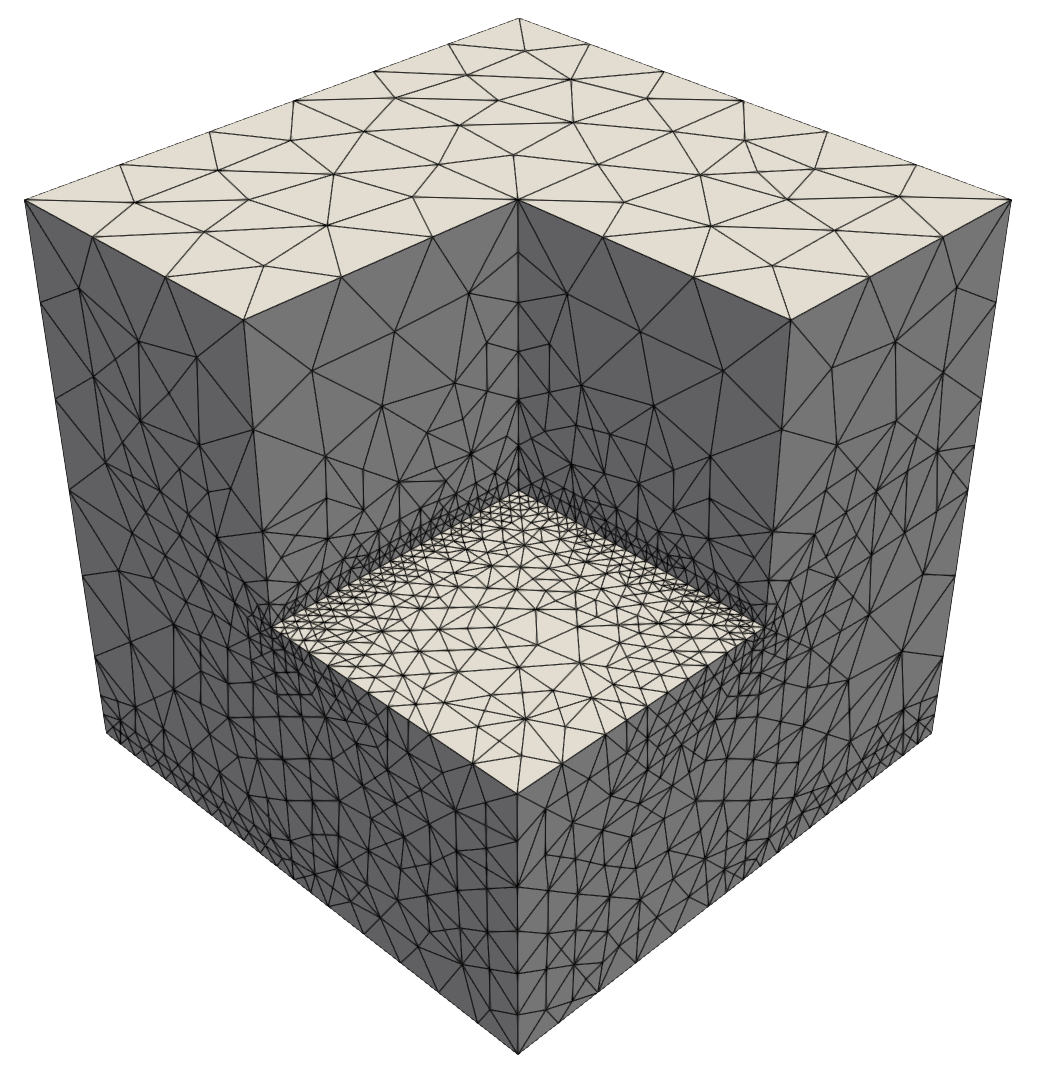}
\end{subfigure}%
\begin{subfigure}{0.33\textwidth}
\centering
\includegraphics[width=.99\linewidth]{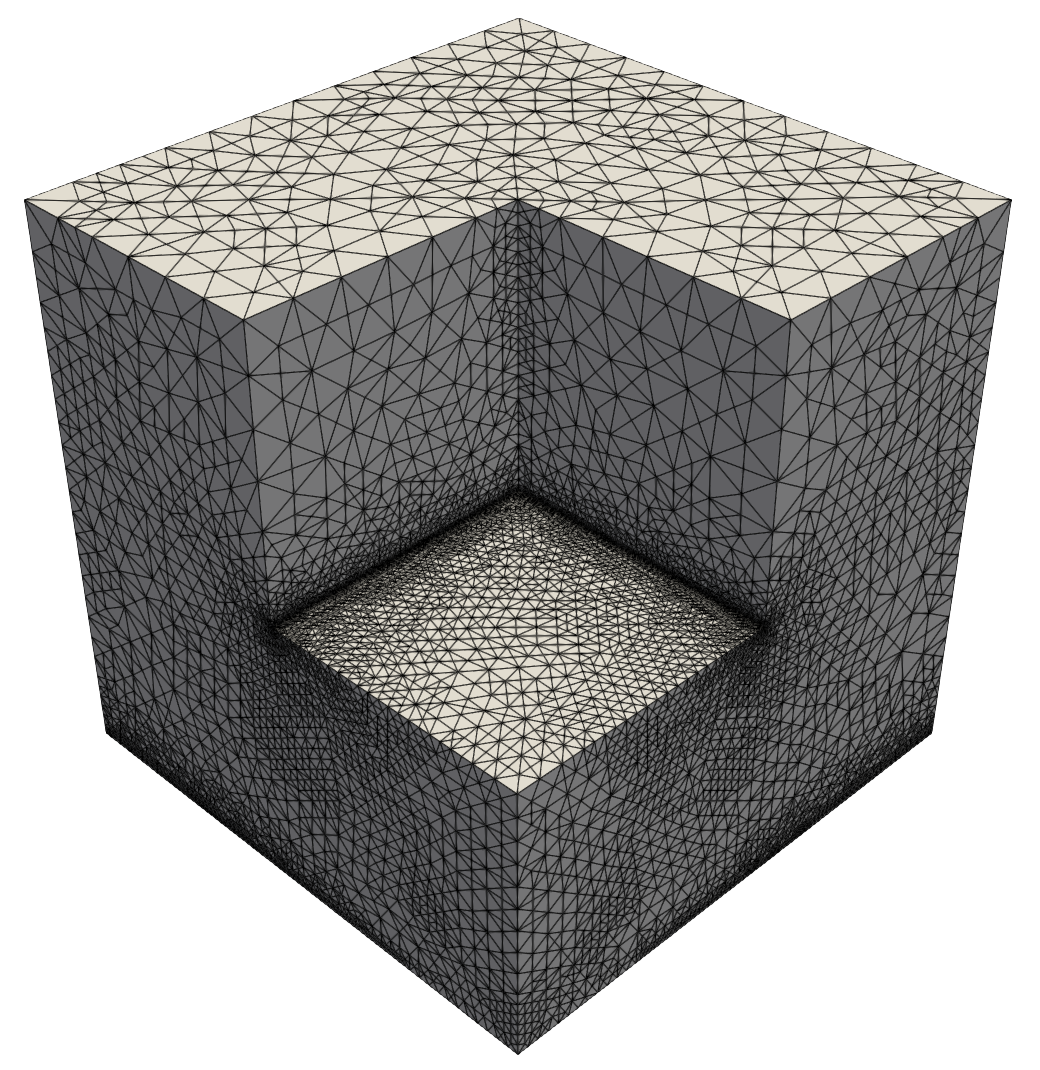}
\end{subfigure}%
\begin{subfigure}{0.33\textwidth}
\centering
\includegraphics[width=.99\linewidth]{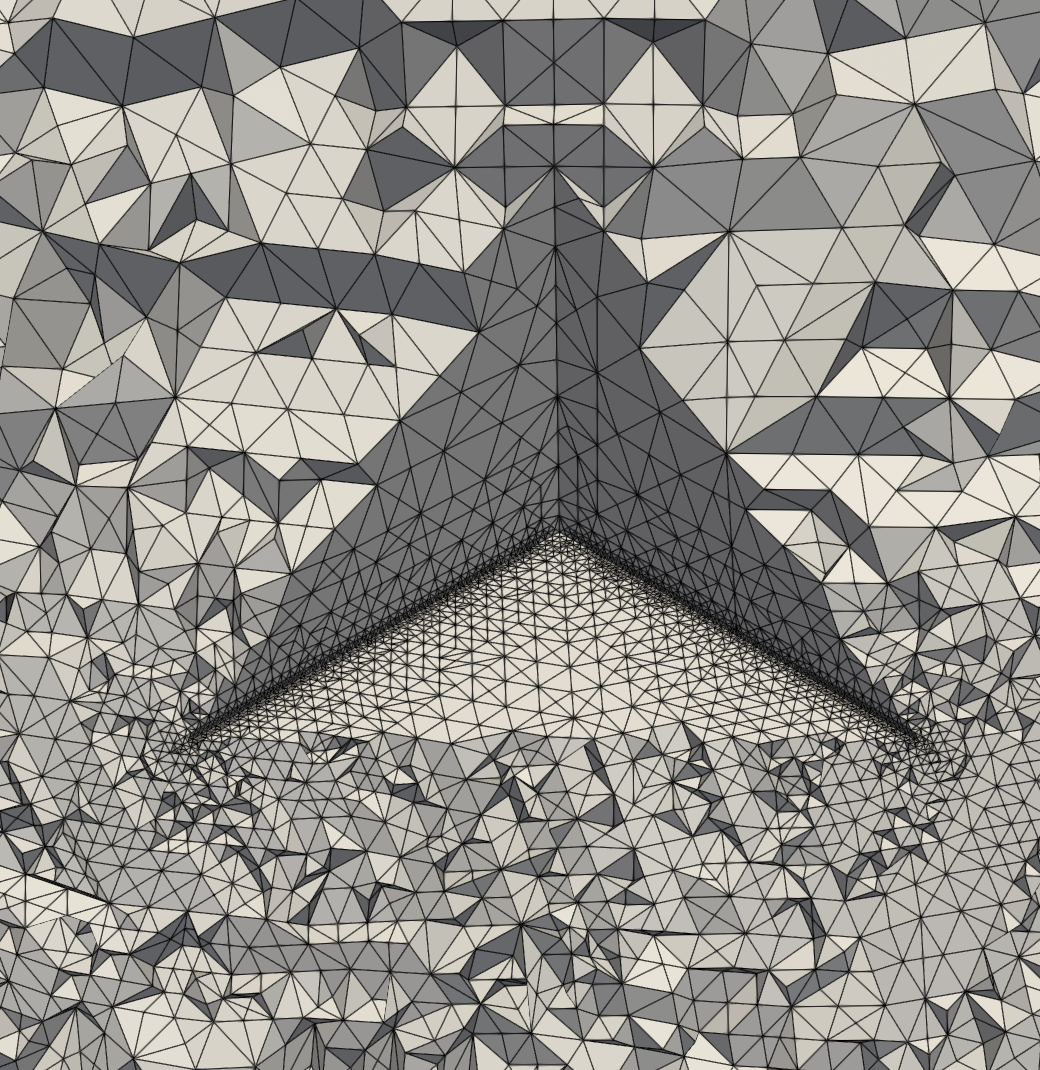}
\end{subfigure}
\caption{Meshes using the $\eta_2$ adative scheme: the adapted mesh at
the fifth adaptive iteration (left), the adapted mesh at the final adaptive
iteration (center), a close-up and cutaway of the final adaptive mesh
(right).}
\label{fig:elasticity_meshes}
\end{figure}

\section{Conclusions}
\label{sec:conclusions}

In this paper, we have considered goal-oriented \emph{a posteriori}
error estimation for nonlinear QoIs when solving nonlinear PDEs with
continuous Galerkin finite element methods. As a starting point, we have
reviewed the derivation of a traditional two-level, or discrete, adjoint-weighted
residual  approach to approximate QoI discretization errors. Inherent in this derivation
is the linearization of both the nonlinear PDE residual and the nonlinear functional of
interest, where higher-order terms are subsequently disregarded. We refer to these
higher-order terms as linearization errors, as they introduce an error into the error estimate.
By considering these linearization errors, we have derived a novel
two-level adjoint-based error estimate that exactly represents the
QoI discretization error between two function spaces at the expense
of increased computational complexity.

We have motivated our interest in this novel error estimate with several
\emph{a priori} justifications and, in particular, have reasoned that the
traditional adjoint-weighted residual error estimate can under-predict
the QoI discretization error for nonlinear QoIs. We have demonstrated that
this under-prediction can in fact occur for two example problems with a
nonlinear Poisson's equation in two spatial dimensions and that
this behavior can be QoI dependent. For these problems,
we have demonstrated that the newly proposed error estimate
does not suffer from this under-prediction and is effective for each
investigated QoI. Additionally, we have demonstrated that the newly proposed
error estimate is more effective than the traditional adjoint-weighted
residual estimate for a finite deformation elasticity example in three
spatial dimensions.

The solution of a nonlinear scalar problem is required for the newly proposed
error estimate. We have provided the analytic solution to this problem when
quadratic QoIs are considered and have outlined how Newton's method can be
used to solve the general nonlinear scalar problem. We have demonstrated
heuristically that the analytic solution for quadratic QoIs is a good initial
guess for Newton's method when more general nonlinear QoIs are considered.

We have shown how the newly proposed error estimate can be localized
and incorporated into an iterative adaptive mesh algorithm. We have
demonstrated that the newly proposed estimate can lead to more
accurate approximations of a nonlinear QoI with fewer degrees of freedom
when compared to uniform refinement and traditional adjoint-based
approaches by considering a nonlinear Poisson's example problem and a finite
deformation elasticity example problem.

Orthogonally, we have demonstrated that the evaluation of the PDE residual
linearization error can lead to a convenient verification check for the
adjoint solution used in traditional adjoint-weighted residual error
estimates.

As an avenue for future work, we have suggested incorporating known recovery
techniques into the novel estimate to reduce its computational complexity.
One could also investigate the effect that the choice of coarse
and fine spaces has on the newly proposed error estimate, as we have only
considered piecweise linear and quadratic Lagrange basis functions, respectively,
in the present work.

\section{Acknowledgments}
\label{sec:ack}

Supported by the Advanced Simulation and Computing program at Sandia National
Laboratories, a multimission laboratory managed and operated by National
Technology and Engineering Solutions of Sandia LLC, a wholly owned subsidiary
of Honeywell International Inc. for the U.S. Department of Energy's National
Nuclear Security Administration under contract DE-NA0003525. This article
describes objective technical results and analysis. Any subjective views or
opinions that might be expressed in the article do not necessarily represent
the views of the U.S. Department of Energy or the United States Government.


\appendix

\section{QoI Values for the Manufactured Solution}
\label{sec:qoi_values_manufactured}

\noindent The QoI values for the manufactured solution and domain defined
in Section~\ref{ssec:manufactured} are:
\begin{equation}
\begin{aligned}
\qoi_1(u) &= \frac{64 \pi^2 (e^{\nicefrac52}-1)^2 (e^5 + e^{\nicefrac52} + 1)}
{e^5(16 \pi^2 + 25)^2} \approx 2.57052599061823, \\
\qoi_2(u) &=  \frac{65536 \pi^6(e^{15} + e^{\nicefrac{45}{4}}
+ e^{\nicefrac{15}{4}} + 1)}{9 e^{\nicefrac{15}{2}} (256 \pi^4 + 4000 \pi^2 +
5625)^2} \approx 2.64090163593838, \\
\qoi_3(u) & = \frac{32 \pi^4(e^{\nicefrac52} - 1)^2 (e^5 + e^{\nicefrac52} +
1)}{25 e^5 (16 \pi^2 + 25)} \approx 9.28106693871883 \times 10^{1}, \\
\qoi_4(u) &\approx 5.67945022.
\end{aligned}
\end{equation}
Here $\qoi_4(u)$ was approximated numerically using Mathematica.

\section{QoI Values for the Solution with Gradient Singularities}
\label{sec:qoi_values_singular}

\noindent The QoI values for the problem definition defined in
Section~\ref{ssec:singular} are:
\begin{equation}
\begin{aligned}
\qoi_1(u) &\approx 6.540644835, \\
\qoi_2(u) &\approx 1.238067612 \times 10^{1}, \\
\qoi_3(u) &\approx 1.596007278 \times 10^{2}, \\
\qoi_4(u) &\approx 9.391778787.
\end{aligned}
\end{equation}

\end{document}